\documentclass[12pt,reqno]{amsart}

\usepackage{amsmath,amsfonts,amsbsy,amsthm}
\usepackage{amssymb}
\usepackage{dsfont}
\usepackage{mathtools}

\usepackage{xy}
\xyoption{matrix} \xyoption{arrow} \xyoption{arc} \xyoption{color}

\usepackage{enumerate}
\usepackage{enumitem}
\setlist{leftmargin=*}

\usepackage{tikz}
\usetikzlibrary{decorations.pathmorphing,shapes,arrows}

\numberwithin{equation}{section}

\makeatletter
\newtheoremstyle{corsivo}
   {\medskipamount}{\medskipamount}%
   {\itshape}{}%
   {\bfseries}{}%
   { }
   {\thmname{#1}\thmnumber{\@ifnotempty{#1}{ }\@upn{#2}}%
    \thmnote{ {\bfseries\boldmath(#3)}}.}%
\makeatother

\theoremstyle{corsivo}
\newtheorem{theorem}{Theorem}[section]
\newtheorem{lemma}[theorem]{Lemma}
\newtheorem{corollary}[theorem]{Corollary}
\newtheorem{proposition}[theorem]{Proposition}

\makeatletter
\newtheoremstyle{dritto}
   {\medskipamount}{\medskipamount}%
   {\rmfamily}{}%
   {\bfseries}{}%
   { }
   {\thmname{#1}\thmnumber{\@ifnotempty{#1}{ }\@upn{#2}}%
    \thmnote{ {\bfseries\boldmath(#3)}}.}%
\makeatother

\theoremstyle{dritto}
\newtheorem{definition}[theorem]{Definition}
\newtheorem{remark}[theorem]{Remark}

\newtheorem{assumption}[theorem]{Assumption}

\newcommand{\sub}[1]{_{\mathrm{#1}}}
\newcommand{\subm}[2]{_{\mathrm{#1},#2 }}
\newcommand{\su}[1]{^{\mathrm{#1}}}

\newcommand{\eps}{\varepsilon}
\newcommand{\epsi}{\varepsilon}

\newcommand{\Id}{\mathds{1}}  

\newcommand{\eu}{\mathrm{e}}
\newcommand{\iu}{\mathrm{i}}   
\newcommand{\di}{\mathrm{d}}

\newcommand{\N}{\mathbb{N}}
\newcommand{\Z}{\mathbb{Z}}
\newcommand{\R}{\mathbb{R}}
\newcommand{\C}{\mathbb{C}}

\newcommand{\Do}{\mathcal{D}}

\newcommand{\Hi}{\mathcal{H}}
\newcommand{\Hf}{\mathcal{H}\sub{f}}
\newcommand{\Df}{\mathcal{D}\sub{f}}
\newcommand{\U}{\mathcal{U}}
\newcommand{\X}{\mathcal{X}}
\newcommand{\B}{\mathcal{B}}
\newcommand{\PO}{\mathcal{P}}

\newcommand{\UZ}{\U\sub{BF}}   

\newcommand{\norm}[1]{\left\| #1 \right\|}

\newcommand{\set}[1]{ \left\{  #1 \right\}}

\DeclareMathOperator{\Tr}{Tr}         

\DeclareMathOperator{\re}{Re} \DeclareMathOperator{\im}{Im}

\DeclareMathOperator{\Ran}{Ran} 
\DeclareMathOperator{\Span}{Span}

\newcommand{\ie}{{\sl i.\,e.\ }}   
\newcommand{\eg}{{\sl e.\,g.\ }} 

\newcommand{\ve}[1]{\mathbf{#1}}
\newcommand{\V}[1]{\mathbf{#1}}

\newcommand{\abs}[1]{\left\lvert#1\right\rvert}
\newcommand{\FC}{{\mathcal{C}}}

\usepackage{comment}

\newcommand{\virg}[1]{``#1''}

\newcommand{\half}{\mbox{\footnotesize $\frac{1}{2}$}}

\renewcommand{\(}{\left(}
\renewcommand{\)}{\right)}

\newcommand{\E}{{\mathrm{e}}}
\newcommand{\I}{\mathrm{i}}
\newcommand{\D}{\mathrm{d}}
\newcommand{\Or}{{\mathcal{O}}}

\let\oldfootnote\footnote
\renewcommand{\footnote}[1]{\oldfootnote{\  #1}}

\title[A new approach to transport coefficients in the QSHE]
{A new approach to transport coefficients\\[2mm]  in the quantum spin Hall effect}
\author[G.~Marcelli, G.~Panati, S.~Teufel]{Giovanna Marcelli \and Gianluca Panati \and Stefan Teufel}
\date{April 2, 2020. Version submitted to \textsl{arxiv.org}.}

\usepackage{hyperref}

\begin{document}

\begin{abstract}
We investigate some foundational issues in the quantum theory of spin transport,
in the general case when the unperturbed Hamiltonian operator $H_0$ does not 
commute with the spin operator in view of Rashba interactions, as in the typical models for the Quantum Spin Hall effect.

A gapped periodic one-particle Hamiltonian $H_0$ is perturbed by adding a constant electric field of intensity $\eps \ll 1$ in the $j$-th direction, and the linear response in terms of a $S$-current in the $i$-th direction is computed,
where $S$ is a generalized spin operator.
We derive a general formula for the spin conductivity that covers both the choice of the conventional and of the proper spin current operator.  We investigate the independence of the spin conductivity from the choice of the fundamental cell 
(\emph{Unit Cell Consistency}), and we isolate a subclass of discrete periodic models where the conventional and the proper
$S$-conductivity agree, thus showing that the controversy about the choice of the spin current operator is immaterial as far as 
models in this class are concerned.   As a consequence of  the general theory, we obtain that whenever the spin is (almost) conserved, the spin conductivity is (approximately) equal to the spin-Chern number.

The method relies on the characterization of a \emph{non-equilibrium almost-stationary state} (NEASS), which well approximates the physical state of the system (in the sense of space-adiabatic perturbation theory) and allows moreover to compute the response of the adiabatic $S$-current as the trace per unit volume of the $S$-current operator times the NEASS. This technique can be applied in a general framework, which includes both discrete and continuum models.
\end{abstract}

\maketitle

\tableofcontents

\newpage

\section{Introduction and main results}



The aim of this paper is to shed some light on the theory of spin transport in gapped (non-interacting) fermionic systems,
a problem which is highly relevant to the research on 
topological insulators (see the end of Section 1.2). \newline
The theory of spin transport, as compared to charge transport, is still in a preliminary stage.
First, despite two decades of scientific debate, no general consensus has been reached yet
about the correct form of the operator representing the spin current density.  
Denoting by $H_0$  the unperturbed Hamiltonian operator, 
by $\ve X=(X_1, \ldots, X_d)$ the position operator 
and by $S_z$ the operator representing the $z$-component of the spin,  one may consider
\footnote{All over the paper we use Hartree units, so that the reduced Planck constant $\hbar$, the mass of the electron 
$m\sub{e}$ and the charge of the positron $e$ are equal to $1$.
With this choice, both the unit of charge conductivity $\frac{e^2}{h}$ and of spin conductivity $\frac{e}{2\pi}$ reduce 
to $\frac{1}{2\pi}$.
}\  
\begin{enumerate} [label=(\roman*), ref=(\roman*)]
\item \label{item: conv S current}
the \virg{conventional} spin current operator
\begin{equation} \label{J_conv}
\ve J \sub{conv}^{S_z}:= 
\half \( \iu [H_0,\ve X] \, S_z + \iu  S_z \, [H_0,\ve X] \)
\end{equation}
which has been used \eg in \cite{Sinovaetalii, ShengShengTingHaldane2005};
\item \label{item: prop S current} the \virg{proper} spin current operator
\begin{equation} 
\label{J_proper}
\ve{J}\sub{prop}^{S_z} := \iu [H_0, \ve X S_z]
\end{equation}
proposed in \cite{ShiZhangXiaoNiu06, cinesi}.
\end{enumerate} 

Whenever $[H_0,S_z]=0$, the two above definitions agree and the theory of spin transport reduces to the theory of charge transport. However, in general $[H_0,S_z] \neq 0$ in topological insulators, as it happens \eg in the model proposed by Kane and Mele in view of the so-called {\it Rashba term} \cite{KaneMele2005a, HasanKane10}. 
As we will explain now (see also \cite{MarcelliPanatiTauber18}), the lack of commutativity poses technical and conceptual problems for the theory of spin transport, and the main objective of our paper is to clarify some of these issues.
 Related to the different possible choices for the spin current operators, 
 first note that whenever $H_0$ is periodic, $\ve J \sub{conv}^{S_z}$ is periodic while $\ve{J}\sub{prop}^{S_z}$ is not, 
which \virg{{\it leads to technical difficulties, but also questions its physical relevance}} \cite{Schulz-Baldes2013}.
Moreover,  in most of the tight-binding models, thanks to the underlying ultraviolet cutoff, $\ve J \sub{conv}^{S_z}$ provides a bounded operator, while $\ve{J}\sub{prop}^{S_z}$  is generically unbounded. On the other hand,  as emphasized in \cite{{ShiZhangXiaoNiu06, cinesi}}, $\ve{J}\sub{prop}^{S_z}$ yields a spin current density associated with a \emph{mesoscopic} sourceless continuity equation and to Onsager relations, in contrast to~$\ve J \sub{conv}^{S_z}$. 


As a second issue, whenever $[H_0,S_z]=0$ the spin conductivity is given, in analogy with charge transport,
by a double commutator formula, namely
\begin{equation} \label{sigma_double_comm}
\sigma_{ij}^{S_z}  = \iu \, \tau \( \Pi_0 S_z  \Big[ [X_i, \Pi_0], [X_j, \Pi_0] \Big] \) \,,       
\end{equation}  
where $\tau$ is the trace per unit volume and $\Pi_0$ the Fermi projector of the gapped system. 
Formula \eqref{sigma_double_comm}, 
equivalently rewritten in terms of Bloch orbitals, has been considered as the starting point 
for further analysis of the robustness of the spin conductivity 
\cite{Sinovaetalii,ShengShengTingHaldane2005}, 
or for a mathematical comparison of spin conductivity  and spin conductance \cite{MarcelliPanatiTauber18}.

In this paper we address two  foundational questions in spin transport theory: 
\renewcommand{\labelenumi}{{\rm(Q\arabic{enumi})}}
\begin{enumerate} 
\item  is it possible to {\it derive}  from the first principles of Quantum Theory, 
in the general case $[H_0,S_z]\neq0$, a double commutator formula for the spin conductivity similar to \eqref{sigma_double_comm}?   
\item to which extent is such a formula affected by a different choice of the spin current operator, namely 
$\ve J \sub{conv}^{S_z}$ versus $\ve J\sub{prop}^{S_z}$?  
\end{enumerate} 
Moreover, any formula for spin transport coefficients should satisfy the so-called {\it Unit Cell Consistency} (UCC), 
namely the requirement that any prediction on macroscopic transport must be independent 
of the choice of the fundamental cell \cite{YangHeZheng2020}. 


In order to answer these questions, we reconsider the whole approach to quantum transport theory.

\subsection{Two paradigms for quantum transport}
\label{Sec:TwoParadigms}

The usual paradigm is based on the adiabatic switching-on of the perturbing electric field. More specifically, one considers the time-dependent Hamiltonian operator
\begin{equation} \label{eqn:H(t)} 
H \sub{switch}(t) := H_0 - f(\eta t) \, \eps \, X_j,
\end{equation}
where $f \colon \R \to [0,1]$ is a smooth function such that $f(s) = 0$ for all $s\le -1$ and $f(s) = 1$ for all $s \ge   0$, \ie the Hamiltonian describes the process where the perturbation is switched on during the finite time interval $[-1/\eta,0]$, for $\eta>0$. 
As $\eta \rightarrow 0^+$, the process becomes adiabatic. 
One assumes that the system is prepared,  at some time {$t \le -1/\eta$}, in the equilibrium state $\Pi_0$
and that the switching occurs adiabatically. The state $\rho_{\eps, \eta}(s)$ at macroscopic time $s = \eta t$ is given by the solution to the time-dependent Schr\"odinger equation
\begin{equation} \label{eqn:AdiabaticEvolution}
\begin{cases}
\iu \, \eta \, \frac{\D}{\D s}{\rho}_{\eps,\eta}(s) = [H\sub{switch}(s), \rho_{\eps, \eta}(s)] \\
\rho_{\eps, \eta}(-1) = \Pi_0
\end{cases}
\end{equation}
The linear response coefficient $\sigma_A$ of an extensive observable $A$  is defined by 
comparing the expectation value of $A$ at time $t_* \geq 0$ (when the perturbation is completely switched-on) 
and in the far past (when the system is in the unperturbed equilibrium state). By considering the adiabatic limit, 
one defines $\sigma_A$ by setting 
$$
\lim_{\eta \to 0^+}  \re \, \tau( A \, \rho_{\epsi, \eta}(t_*)) -  \re \, \tau( A \, \Pi_0) =: \epsi \, \sigma_A + o(\epsi)  
\qquad \quad \text{as }\epsi \to 0.
$$
The real part appears in the formula since one does not know {\it a priori} whether the conditional cyclicity of the trace per unit volume can be invoked.
\footnote{A similar phenomenon appears in Quaternionic Quantum Mechanics, where the Hilbert space trace fails to be cyclic \cite{MorettiOppio18}. 
}\quad 
The standard approach for  obtaining a tractible fomula for $\sigma_A$ is to first   approximate   $\rho_{\epsi, \eta}(0)$ by first-order time-dependent perturbation theory, and then to formally exchange  the small field limit and the adiabatic limit, see \eg \cite{AizenmanGraf98,Teufel20}. Choosing 
 $f(\eta t) = \eu^{\eta t} \chi_{(- \infty, 0]}(t) + \chi_{(0, + \infty)}(t)$, this results in
  Kubo's formula \cite{Kubo57} for the linear response coefficients\footnote{Note that the specific choice $f(\eta t) = \E^{\eta t}$ for $t\leq 0$ has the computational advantage that the integral in \eqref{KuboFormula} becomes the inverse Liouvillian, \ie that the right hand side of \eqref{KuboFormula}  for finite $\eta>0$ equals, at least formally, 
$\tau((\mathcal{L}_{H_0} - \I \eta)^{-1}([X_j,\Pi_0]) A)$. However, in the adiabatic limit
any other integrable and smooth choice for the switching function $f$ leads to the same value for $ \sigma_A^{\rm Kubo}$.} 
\begin{equation}\label{KuboFormula}
 \sigma_A^{\rm Kubo} := -\I \lim_{\eta\to  {0^+}}  { \int\limits_{-\infty}^0}\D t  \,\E^{\eta t}  \,\tau 
 \Big(  \E^{-\I H_0  t}\, \left[X_j,\Pi_0\right]  \,\E^{\I H_0  t} A   \Big)  \,.
\end{equation}
In the case of charge transport, one considers the response of a   charge current in the $i$-th direction, $i \in \set{1, \ldots, d}$, whose corresponding quantum mechanical operator is
$ J_i\su{c} := \iu [H_0, X_i]\,,$
and from \eqref{KuboFormula} one obtains the  formula 
\begin{equation} \label{Kubo_charge}
 \sigma_{J_i\su{c}}^{\rm Kubo} = \I \,\tau \Big( \Pi_0 \Big[ [X_i, \Pi_0], [X_j,\Pi_0] \Big]  \Big)\,.
\end{equation}
The importance of the double commutator formula \eqref{Kubo_charge} 
(sometimes dubbed {\it Kubo-Chern formula}) cannot be overstated, 
as it implies \eg quantization of Hall conductivity in $2$-dimensional systems \cite{AvronSeilerSimon, Bellissard94, KleinSeiler90, Graf07}. When considering  spin transport, we had to face the fact that even the algebra which leads formally to 
\eqref{Kubo_charge} becomes cumbersome for spin currents, whenever $[H_0, S_z] \neq 0$. 
Moreover,  the fact that the formula is intrinsic (\ie does not depend on the choice of the switching function appearing in \eqref{eqn:H(t)} and on the choice of $t_* \geq 0$) is not obvious as far as spin currents are concerned.


Thus we propose an alternative way of computing linear response coefficients based on the non-equilibrium almost-stationary states (NEASS), a concept related to the almost-invariant subspaces in space-adiabatic perturbation theory \cite{PanatiSpohnTeufel03, PanatiSpohnTeufel03b, Teufel03}.
In a nutshell, the NEASS  $\Pi^\epsi$ is the unique  almost-invariant state for the perturbed 
stationary Hamiltonian  $H^\epsi = H_0 - \epsi X_j$ that is $\epsi$-close to the  equilibrium state $\Pi_0$ of the unperturbed Hamiltonian $H_0$. More precisely, the asymptotic expansion of $\Pi^\epsi$ in powers of $\epsi$  is uniquely determined by the conditions
\begin{itemize}
\item[(i)] $ [\Pi^\epsi, H^\epsi] = \Or (\epsi^\infty)$
\item[(ii)] $\Pi^\epsi - \Pi_0 = \Or(\epsi)$\,,
\end{itemize}
where details on the precise norms will be given later. Assuming, for the moment,  that the state of the system at times when the perturbation has been turned on is approximately given by the NEASS $\Pi^\epsi$, we find a simple prescription for computing linear (and also higher order) response coefficients:  let $\Pi^\epsi = \Pi_0 + \epsi \Pi_1 + o(\epsi)$, then from 
$$
 \tau(A \, \Pi^\epsi ) -  \tau(A \, \Pi_0)
= \epsi  \tau(A \, \Pi_1) + o(\epsi)     
$$
one concludes that
\begin{equation}\label{LinearResponseNew}
\sigma_A = \re \tau(A \, \Pi_1)\,.
\end{equation}
In this paper we will show how to compute formulas for the spin-conductivities based on  formula \eqref{LinearResponseNew} for linear response coefficients, instead of \eqref{KuboFormula}.
The advantage of this method is that the operator $\Pi_1$ is rather explicit, namely  
$\Pi_1 = \mathcal{I}\( \overline{[  X_j,\Pi_0]} \)$, where the overline denotes the operator closure 
and $\mathcal{I}$ is the inverse of the Liouvillian operator  $ B \mapsto [H_0, B]$, with integral representation \eqref{eqn:I(A)}.

Of course one expects, and formally it is also easy to see, that the two expressions 
\eqref{KuboFormula} and \eqref{LinearResponseNew} agree. However, in the present setting -- 
where expectations are obtained via a trace per unit volume which is only conditionally cyclic -- this is not straightforward to prove. Moreover, both formulas are somewhat heuristic: for \eqref{KuboFormula} we assumed applicability of time-dependent perturbation theory also for long adiabatic time-scales, while for \eqref{LinearResponseNew} we just postulated that the perturbed system is in the state $\Pi^\epsi$.

In order to reconcile and justify both approaches, one needs to prove that in the adiabatic regime the dynamical  switching 
drives the initial equilibrium state $\Pi_0$ approximately into the NEASS $\Pi^\epsi$, \ie that
 the state $\rho_{\epsi, \eta}(t )$ is close to $\Pi^\epsi$. 
 Indeed, it is shown in \cite{MarcelliTeufel19} that for times $t \geq 0$ and any $n,m\in \N^*$ 
 \begin{equation} \label{Intro:approx}
\sup_{\eta\in I_{m,\epsi}} \left| \tau( A \, \rho_{\eps, \eta}(t)) - \tau( A \, \Pi^\eps) \right| = \Or(\epsi^n)   
\qquad\qquad (\text{for } t  \geq 0)
\end{equation}
uniformly on bounded intervals in (macroscopic) time. Here $I_{m,\epsi} = [\epsi^m ,\epsi^{1/m}]$ 
is an intervall of admissible time-scales for the switching. Too slow switching ($\eta\ll \epsi^m$ for all $m\in \N^{*}$) 
must be excluded, because due to tunneling the NEASS decays on such long times-scales, while too fast switching ($1\gg \eta\gg \epsi^{1/m}$ for all $m\in \N^*$) would merely yield an error $o(1)$ on the right hand side of \eqref{Intro:approx}.
 
In other words, the initial equilibrium state $\Pi_0$ dynamically evolves into the {NEASS} independently of the shape of the switching function up to lower order errors. A proof of a similar statement in the context of interacting models on lattices is provided in \cite{Teufel19} and the issue of justifying linear response and Kubo's formula is briefly reviewed in \cite{Teufel20}.

\subsection{Main results on spin transport and conductivity}
\label{Sec:Main results} 

By using the NEASS paradigm, we will answer the questions ($Q_1$) and ($Q_2$) stated before,
at least in the periodic setting. Let us shortly summarize the main results in the paper. 

We consider a crystalline system of non-interacting fermions, whose one-body Hamiltonian $H_0$ is periodic. 
This operator acts on the Hilbert space $\Hi = L^2(\X) \otimes \C^N$, where either $\X=\R^d$ (continuum case) or  $\X\subset\R^d$ is a discrete set  (discrete case), and $N$ is the number of internal degrees of freedom of the particle, which may include spin; periodicity of $H_0$ is understood with respect to (magnetic) translations along vectors in a Bravais lattice $\Gamma \simeq \Z^d $. We assume that the Hamiltonian $H_0$ has a spectral gap, and that the initial state of the system is given by the spectral projection $\Pi_0$ on the bands below this gap (Fermi projector). The system is  driven out of equilibrium by applying a constant electric field of intensity $\eps \ll 1$ pointing in the $j$-th direction, $j \in \set{1,\ldots,d}$. Hence, the stationary Hamiltonian of the perturbed system is $H^\eps = H_0 - \eps X_j$, where $X_j$ is the $j$-th component of the position operator.


We consider a generalized spin operator in the form $S = \Id_{L^2(\X)} \otimes s$ and -- 
denoting by $J\subm{conv}{i}^S:=\half \( \iu [H_0, X_i] \, S + \iu  S \, [H_0, X_i] \)$ and $J\subm{prop}{i}^S = \iu [H_0, X_i S]$  the corresponding \emph{conventional} and \emph{proper} $S$-current operator -- 
we define the \emph{conventional} and \emph{proper} $S$-conductivity, respectively, as
\footnote{Notice that the $j$-dependence of $\sigma\subm{conv/prop}{ij}^S$ is hidden on the right-hand side of the following definition in $\Pi_1$,  see its definition in Proposition~\ref{prop:Pi1}\ref{item:propSA3}.} 
\begin{align}
\label{eqn:defn conv/prop sigma}
\re\tau(J\subm{conv/prop}{i}^S\,\Pi^\eps)-  \re \tau(J\subm{conv/prop}{i}^S\,\Pi_0)
=:\eps \, \sigma\subm{conv/prop}{ij}^S+ o(\eps).
\end{align}
 In view of the controversy on the choice of the spin current operator discussed at the beginning of the Introduction, 
we find convenient the decomposition
\begin{equation} 
\label{eqn:splitting JpropPi1}
J\subm{prop}{i}^S\Pi_1= \iu [H_0, X_i S] \Pi_1  
=  {\iu [H_0, X_i] S} \Pi_1 
+  X_i \big( \iu [H_0, S]  \big) \Pi_1  =\mathsf{O}+X_i\mathsf{R} 
\end{equation}
where we have defined the operators 
\begin{equation}
\label{eqn:defn T and R}
\mathsf{O}:=\iu  [H_0,X_i] S\Pi_1\quad\text{ and }\quad\mathsf{R}:=\iu  [H_0, S ]  \Pi_1. 
\end{equation}
In this decomposition, the \emph{$S$-orbital} term $\mathsf{O}$ contains the contribution associated with the conventional $S$-current operator, while the \emph{$S$-rotation} terms $X_i\mathsf{R}$ contains corrections related to the replacement of the latter with the proper $S$-current operator. More precisely, we prove in Theorem~\ref{thm:notKubo} that splitting \eqref{eqn:splitting JpropPi1} leads to
\begin{equation}\label{sigma1+sigma2}   
\sigma\subm{prop}{ij}^S=\sigma\subm{conv}{ij}^S+\sigma\subm{rot}{ij}^S,
\end{equation}
where 
\begin{equation}
\label{Intro:sigma_orbital}
\begin{aligned}
\sigma\subm{conv}{ij}^S  & =  \re \, \tau \Big( \iu \, \Pi_0  \Big[ [X_i,\Pi_0] S,  [X_j,\Pi_0] \Big] \Big) \\
& + \re \, \tau \Big(\iu \, [H_0,X_i\su{D}] S\su{OD} \Pi_1 + \iu \, X_i\su{OD} [S,H_0]\Pi_1 \Big),
\end{aligned}
\end{equation}
with $A\su{D}$ (resp.\ $A\su{OD}$) referring to the diagonal (resp.\ off-diagonal) part of the operator $A$ with respect to the orthogonal decomposition induced by $\Pi_0$, and the \emph{rotation $S$-conductivity} is
\begin{equation}
\label{Intro:sigma_rot}
\begin{aligned}
\sigma\subm{rot}{ij}^S  & = \re \, \tau(X_i \mathsf{R}) = \re \, \tau\Big(\iu X_i  [H_0, S ]  \Pi_1\Big).
\end{aligned}
\end{equation}		
Notice that the first line of \eqref{Intro:sigma_orbital} is in the form of a current-current correlation at the equilibrium, involving the conventional $S$-current and the charge current, while the second line involves $\Pi_1$.
Moreover, the trace per unit volume in \eqref{Intro:sigma_orbital} and \eqref{Intro:sigma_rot} can be replaced with the 
ordinary trace of the operator restricted to the fundamental cell, up to a volume factor, as in the statement of Theorem \ref{thm:notKubo}, even if the operator appearing in \eqref{Intro:sigma_rot} is not periodic.
					
In order to analyze the {$S$-rotation} contribution $\sigma\subm{rot}{ij}^S$, we preliminary prove
in Proposition~\ref{prop:Btorque}\ref{item:doublecomm formula} that for any bounded periodic observable $B$, satisfying suitable regularity properties, the expectation of the 
\emph{$B$-torque operator} $\iu [H_0, B]$ on $\Pi_1$ is given by a double commutator formula, namely
\begin{equation*} 
\tau  \big( \iu [H_0, B] \, \Pi_1 \big)  = \tau \big(  \underbrace{\iu \Pi_0 \, \big[ [ \Pi_0, B] , [\Pi_0, X_2] \big]}_{\mathcal{T}_B} \big),
\end{equation*}
where the operator ${\mathcal{T}_B}$ may be dubbed \emph{$B$-torque response} in agreement with \cite{MarcelliPanatiTauber18}.  If, in addition, $[B,X_j]=0$ then $\tau  \big( \iu [H_0, B] \, \Pi_1 \big)=0$, as stated in Proposition~\ref{prop:Btorque}\ref{item: vanishing of doublecomm formula if A per}. 
Physically, this result means that even if $\iu [H_0, B]\neq 0$, the fact that $B$ commutes with the perturbation $-\eps X_j$ implies the \emph{mesoscopic} conservation of the observable $B$, at least within first order approximation in the NEASS. In particular, when $B=S_z$ (or for any generalized spin operator $S$, see Corollary~\ref{cor:tau R=0}), we have that the expectation of the spin-torque on $\Pi_1$ equals the expectation of the \emph{spin-torque response} $\mathcal{T}_{S_z}$ and that the latter vanishes, in agreement with \cite[Theorem 2.8]{MarcelliPanatiTauber18}. Notice that the vanishing of the expectation of the spin-torque response is a condition singled-out in \cite{MarcelliPanatiTauber18} to obtain the equality of spin conductivity and spin conductance in $2$-dimensional systems. On the other hand, since $X_j$ (resp.\ $X_j S$) is unbounded, Proposition \ref{prop:Btorque} does not  provide any information on charge (resp.\ spin) conductivity, whose analysis requires an additional technical effort. \newline


As a further step, we consider the Unit Cell Consistency (UCC) of both the contributions to the proper $S$-conductivity 
appearing in \eqref{sigma1+sigma2}. We prove in Proposition~\ref{prop:UCC} that $\sigma\sub{conv}^S$ always satisfies UCC, 
while for the additional contribution  $\sigma^S\sub{rot}$
we can prove UCC only if the model enjoys a discrete rotational symmetry, 
in agreement with the claim in \cite{ShiZhangXiaoNiu06} that the use of $\ve{J}\sub{prop}^{S_z}$
is \virg{{\it possible for systems where the spin generation in the bulk is absent due to symmetry reasons.}}
In Proposition \ref{prop:Srotsigmazeroundersym} we isolate a subclass of discrete models, enjoying a discrete rotational symmetry and a further property, such that   $\sigma\sub{conv}^S = \sigma\sub{prop}^S$. 
Remarkably, the paradigmatic model proposed by Kane and Mele is in this class.
A crucial consequence is that, for this class of models, 
the choice of the spin current operator (either $\ve J^S\sub{conv}$ or $\ve J^S\sub{prop}$)
is immaterial as far as the $S$-conductivity is concerned.


While the paper is focused on transport theory, 
one of our long-term goals is to clarify the relation between the spin transport coefficients 
and the topological invariants associated to Quantum Spin Hall (QSH) insulators. 
These materials, theoretically predicted in \cite{KaneMele2005a, KaneMele2005b} and soon experimentally realized 
 \cite{Koenig_et_al08,SinovaValenzuelaWunderlichBackJungwirth15}, 
display dissipationless edge spin currents, which are robust against continuous deformations of the model
and disorder \cite{ShengShengTingHaldane2005}. 
A crucial issue, both for fundamental understanding and for potential applications,  
is whether there exists a bulk topological invariant \virg{protecting} the QSH effect. 
Two candidates have been extensively investigated in the literature.
First, the $\Z_2$-valued index proposed by Fu, Kane and Mele \cite{KaneMele2005b,FuKane2006}, 
whose definition and geometric properties rely on the fermionic time-reversal symmetry of the system
\cite{GrafPorta2013, FiorenzaMonacoPanati2016b, CorneanMonacoTeufel2017}.
Second, the (half-)integer-valued \emph{spin-Chern number}, introduced in \cite{ShengWengShengHaldane2006} 
via spin dependent boundary conditions, and later intrinsically redefined by Prodan as a bulk invariant 
\cite{Prodan2009}, which relies instead on the almost-conservation of spin, and is associated 
to robust spin edge currents \cite{Schulz-Baldes2013, Prodan2008c, Schulz-Baldes15, KatsuraKoma16}.  \newline
Our analysis establishes a direct relation between the bulk spin conductivity and 
the spin-Chern number, in agreement with the (recent) discovery that QSH plateaux may persist under broken time-reversal symmetry \cite{DuKnezSullivanDu2016}. 
Indeed, whenever spin is conserved, our results yield that the (bulk) spin conductivity equals the 
spin-Chern number (Remark \ref{rem:Spin chern number}). Moreover, the result is robust: if spin is approximately conserved, 
with errors of order $\Or(\lambda)$, then the mentioned equality holds true up to a correction of order 
$\Or(\lambda)$ (Proposition \ref{prop:robust spin chern}), in analogy with the persistence of 
\emph{edge} spin currents proved in \cite{Schulz-Baldes2013}. 

\bigskip

  
In summary, our paper contributes to put spin transport theory on a firm mathematical ground: 
We derive a new formula for the spin conductivity which covers both the choice of 
the conventional and the proper spin current operator; we isolate conditions under 
which UCC is satisfied and additional conditions which guarantee that $\sigma\sub{conv}^S = \sigma\sub{prop}^S$;
we make connection with the spin-Chern number. 
We hope that our mathematical investigations will contribute to clarify some of the controversies 
in the emerging and promising field of spintronics, and will stimulate a fruitful exchange of ideas 
between mathematicians and solid state physicists. 
While, for technical reasons, this paper focuses on the case of periodic non-interacting systems, we are confident that our approach can be suitably generalized to random and interacting systems.

\subsection{Further reference to the literature}
\label{Sec:Literature}

We conclude this introduction with some comments on the existing literature. 
The mathematical literature on the quantum Hall effect and on the justification of linear response theory in the context of quantum transport of charge is by now prodigious, and we will not embark in the task of giving a full account of it here. Indeed, several different mathematical problems have been labeled ``proving Kubo's formula''
and a short review highlighting the differences will appear elsewhere \cite{Teufel20}.
Here we only mention a few works without going into any detail: 
 A similar approach to the one we use was employed in \cite{StiepanTeufel13}, where  Kubo's formula for the Hall conductivity of simple isolated bands is derived using semi-classical methods.
The rigorous derivation of Kubo's formula for interacting fermionic systems on the lattice has  recently been done in \cite{BachmannDeRoeckFraas17,MonacoTeufel17, Teufel19}, where \cite{BachmannDeRoeckFraas17,MonacoTeufel17} consider only situations where the perturbation does not close the spectral gap. 
A similar result for non-interacting fermions in the continuum is in preparation \cite{MarcelliTeufel19}, 
generalizing a previous result \cite{ElgartSchlein04} which also assumes a non-closing-gap condition.  
In many other works  Kubo's formula for the Hall conductivity is taken as a starting point and the objective is
to prove quantization of the Hall {\it plateaux}  also in presence of disorder, assuming the Fermi energy lies in a mobility gap \cite{Bellissard94,AizenmanGraf98,BoucletGerminetKleinSchenker05}, or including interaction effects \cite{HastingsMichalakis15,GiulianiMastropietroPorta17, Bachmann_et_al17}, with the aim of proving \emph{universality} of the Hall conductivity. Moreover, the linear response to a quenched perturbation has been recently analyzed in 
\cite{CFLSD20}. 
Finally, in 
\cite{BrudeSiqueiraPedra16,DeNittisLein16} (and references therein) mathematical frameworks are developed, within which the  applicability of linear response theory in very general random resp.\ interacting systems can be established. 
However, a rigorous justification of   Kubo's formula for the quantum Hall conductivity in situations with mobility gap   is still a completely open problem, even in the case of non-interacting systems on the lattice. 

Linear response theory can also be considered in the case of heat or charge fluxes induced by thermodynamical (\ie non-mechanical) driving forces, such as deviations of temperature or chemical potential from their equilibrium values.  In this context, the validity of the \emph{Green-Kubo formula} has been extensively investigated in algebraic quantum statistical mechanics, by relating it to the structure of non-equilibrium steady states \cite{JOP1, JOP2, JOP3, JOP4}. 

The field of spintronics and of quantum transport of spin is relatively new, but has already attracted a lot of attention both in the physics and mathematics communities. Results concerning the quantization and robustness of spin Hall currents in the presence of disorder \cite{Prodan2009, Schulz-Baldes2013} and of interactions \cite{AntinucciMastropietroPorta17, MastropietroPorta17} also rely, to some extent, on a Kubo-like formula. We foresee the possibility of adapting the techniques developed in \cite{MonacoTeufel17, Teufel19} to derive such formulas from first principles also in the context of interacting fermions on a lattice.

A further development of the present line of research consists in pushing the validity of Kubo-like formulas for adiabatic $S$-currents to arbitrarily high orders in the adiabatic parameter $\eps$. This has been achieved in \cite{KleinSeiler90} for quantum Hall systems.

\medskip

\noindent{\bf Acknowledgements.} We are very grateful to Domenico Monaco for intensive discussions and 
exchange of ideas in an early stage of our project. We are thankful to Luca Fresta, Vojkan Jak\v{s}i\'c, Marcello Porta, and Cl\'{e}ment Tauber for intensive and fruitful discussions. S.~T. also benefited from useful and stimulating interactions with Giuseppe De Nittis and Max Lein. Financial support from the German Science Foundation within the Research Training Group 1838 on ``Spectral theory and dynamics of quantum systems''  is gratefully acknowledged.


\newpage

\section{Periodic operators and trace per unit volume} 
\label{sec:TPUV}

In condensed matter physics it is customary to describe  crystalline solids by means of periodic Hamiltonian operators. The appropriate trace-like functional used to compute thermodynamic expectations of periodic observables is given by the trace per unit volume. This Section is devoted to recall some generalities about this framework.

Let $\X$ denote the configuration space of a $d$-dimensional crystal. We will treat both \emph{continuum models}, in which $\X=\R^d$ equipped with the Lebesgue measure, and \emph{discrete models}, in which $\X\subset\R^d$ is a discrete set of points arranged in a crystalline structure, equipped with the counting measure (in $d=2$ think of the square lattice $\Z^2$ or of the honeycomb structure, for example). In general ``crystalline structure'' means that we assume the existence of a Bravais lattice
\begin{equation}
\label{eqn:Gamma}
\Gamma = \Span_\Z \set{a_1, \ldots, a_d} \simeq \Z^d
\end{equation}
 that acts on $\X$ by translations, \ie $\mathrm{T}_\gamma x:= x+\gamma$ for $\gamma\in\Gamma$ defines a group action $\mathrm{T}:\Gamma\times\X\to\X$.

We consider the one-particle Hilbert space 
\[
\Hi = L^2(\X) \otimes \C^N \simeq L^2(\X ,\C^N )
\]
for a particle moving on $\X$ and having  $N$   internal degrees of freedom (\eg spin). 
In the following we will write elements of $\Hi$ as $\C^N$-valued functions on $\X$.
We assume that there is a unitary representation $T$ of $\Gamma$  on $\Hi$ by  
\emph{(magnetic) translation operators}
\begin{equation} \label{eq:defn transl}
 (T_\gamma \psi)(x )   :=  M(\gamma,x) \psi(x-\gamma ), \quad\mbox{for all $\gamma\in\Gamma$ and $\psi\in\Hi$,}
\end{equation}
where $M: \Gamma\times \X \to \U(\C^N)$ are unitaries satisfying the cocycle condition\footnote{The case of \emph{magnetic translations} \cite{Zak64} is included in this framework and thus the Bloch--Landau Hamiltonian can be considered in our setting, assuming a rationality condition on the magnetic flux per unit cell.}
\[
M(\gamma_1+\gamma_2,x) = M(\gamma_2,x) M(\gamma_1,x-\gamma_2) \quad\mbox{for all $\gamma_1,\gamma_2\in\Gamma$ and $x\in\X$.}
\]
%
\emph{Position operators} for $j \in  \: \set{1, \ldots, d}$ are defined via
\begin{equation}
\label{eqn:def position op}
 (X_j \psi)(x ) := x_j \psi(x ), \quad\mbox{for all $\psi\in \Do(X_j)$.}
\end{equation}
An operator $A$ on $\Hi$ is called \emph{periodic} or, more specifically, \emph{$\Gamma$-periodic} if $[A, T_\gamma] = 0$ for all $\gamma \in \Gamma$. The following simple observation is very useful.

\begin{lemma} \label{lemma:derivata}
For any   periodic operator $A$, the operator $[A, X_j]$ is also  periodic.
\end{lemma}

\noindent

Notice that, in general, the operator $[A, X_j]$ might be non-densely defined or even defined on the trivial subspace $\set{0}$, as pointed out in \cite[III-\S 5.1]{Kato66}. This pathology will not appear for the specific operators we will consider in the following Sections.

\begin{proof}
Noticing that 
\begin{eqnarray*}
(X_j T_\gamma   \psi)(x) &=& x_j M(\gamma,x) \psi(x-\gamma ) = M(\gamma,x)(x-\gamma)_j \psi(x-\gamma)  + \gamma_j M(\gamma,x) \psi(x-\gamma ) \\ &=&
(T_\gamma X_j\psi)(x) + \gamma_j (T_\gamma \psi)(x)\,,
\end{eqnarray*}
we get
$[X_j, T_\gamma] = \gamma_j T_\gamma$. Thus,  by the Jacobi identity,
\[ [[A,X_j],T_\gamma] = - [[X_j, T_\gamma], A] - [[T_\gamma,A],X_j] = -\gamma_j [T_\gamma,A] - [[T_\gamma,A], X_j]\,, \]
which vanishes since $[T_\gamma,A]=0$ by assumption.
\end{proof}

The analysis of periodic operators is best performed in the so-called (magnetic)  
Bloch--Floquet--Zak representation (see \eg 
{\cite{Kuchment16, MonacoPanatiPisanteTeufel16, PanatiPisante13}} and references therein). 
The {\emph{(magnetic) Bloch--Floquet--Zak transform}} is initially defined
on compactly supported functions $\psi\in C_0 (\X,\C^N)\subset L^2(\X,\C^N)$ as
\begin{equation}
\label{eqn:defn UZ}
(\UZ \psi)(k,y) :=\eu^{-\iu k \cdot y} \sum_{\gamma \in \Gamma} \eu^{\iu k \cdot \gamma} (T_\gamma \psi)(y)\qquad k\in\R^d,\, y\in \X.
\end{equation}
By construction, for fixed $k\in \R^d$, the function $(\UZ \psi)(k,\cdot)$ is periodic with respect to the magnetic translations \eqref{eq:defn transl}, hence it defines an element in the Hilbert space 
$$
\Hf := \{\varphi \in L^2\sub{loc}(\X,\C^N)\,|\,  T_\gamma \varphi=\varphi \mbox{ for all $\gamma\in\Gamma$} \} \quad\mbox{ with } \quad \|\varphi\|_{\Hf}^2 := \int_{\FC_1}\D y\, |\varphi(y)|^2,    
$$
where the norm refers to a fundamental cell  $\FC_1$ for $\Gamma$ (see \eqref{eqn:defn FCL}). 
As functions of $k$, elements in the range of $\UZ$ are not periodic with respect to the reciprocal lattice  $\Gamma^*$, but rather $\varrho$-equivariant, namely
\[
(\UZ \psi)(k+\gamma^*, y) = \varrho(\gamma^*) (\UZ \psi)(k , y)   \text{ for all } \gamma^* \in \Gamma^*,
\]
where
\footnote{We denote by $\U(\Hf)$ the group of the unitary operators on $\Hf$.
} 
\begin{equation}
\label{eqn:rho representation}
\varrho\colon \Gamma^*\to \U(\Hf), \qquad(\varrho(\gamma^*)\varphi)(y):=\E^{-\I \gamma^*\cdot y} \varphi(y),
\end{equation}
defines a unitary representation of $\Gamma^*$ 
on $\Hf$.  The map defined by \eqref{eqn:defn UZ} extends to a unitary operator
\[ 
\UZ \colon \Hi \to \Hi_\varrho, 
\]
where $\Hi_\varrho \equiv L^2_{\varrho}(\R^d,\Hf)$ is the space of locally-$L^2$, $\Hf$-valued, $\varrho$-equivariant functions on $\R^d$. Denoting by $\mathbb{B}^d$ a fundamental domain for $\Gamma^*$, the inverse transformation 
$\UZ^{-1} \colon \Hi_\varrho\to\Hi$, 
sometimes dubbed {\it Wannier transform}, 
is explicitly given by
$$
(\UZ^{-1} \varphi)(x) =\frac{1}{\abs{\mathbb{B}^d}}\int_{\mathbb{B}^d}\di k\,\E^{\iu k\cdot x}\varphi(k,x).
$$
At least formally, a periodic operator $A$ on $\Hi$ becomes a \emph{covariant fibered operator} on $\Hi_\varrho$.
More precisely, taking into account the following inclusion and natural isomorphism 
\begin{equation} \label{eq:inclusion}
L^2_{\varrho}(\R^d,\Hf) \subset L^2(\R^d,\Hf) \simeq \int_{\R^d}^{\oplus} \di k\,\Hf ,
\end{equation}
one has
$$
\UZ \, A \, \UZ^{-1} = \int_{\R^d}^{\oplus}\di k\, A(k),
$$
where each $A(k)$ acts on $\Hf$ 
and satisfies the covariance property $A(k+\gamma^*) = 
\varrho(\gamma^*) \, A(k) \, \varrho(\gamma^*)^{-1}
$ for all $k\in\R^d$ and $\gamma^*\in\Gamma^*$.

Most relevant extensive observables in crystalline systems are  periodic self-adjoint operators.  However, 
in an infinite system neither these periodic extensive observables nor translation invariant states are trace class. The appropriate functional is instead given by the \emph{trace per unit volume} $\tau$, which is well suited to take into account invariance or covariance by discrete lattice translations in the setting of periodic or more generally ergodic operators {\cite{Bellissard86, PasturFigotin92, BoucletGerminetKleinSchenker05, AizenmanWarzel15}}. The trace per unit volume is defined as follows (compare \cite[Prop.~3.20]{BoucletGerminetKleinSchenker05}).
Denote by $\chi_\Omega$ the orthogonal projection on $\Hi$ which multiplies by the characteristic function of $\Omega \subset \X$.
For any $L\in 2\N+1$, we set 
\begin{equation}
\label{eqn:defn FCL}
\FC_L:=\left\lbrace x\in \X : x = \sum_{j=1}^d \alpha_j \, a_j \text{ with }  |\alpha_j|\leq L/2 \; \forall\: j \in \set{1,\ldots,d}\right\rbrace
\end{equation}
and $\chi_L:=\chi_{\FC_L}$. The set $\FC_1$ is called a \emph{fundamental or primitive (unit) cell}. It is not unique since the choice of the spanning vectors $\{a_j\}_{1\leq j\leq d}$ for $\Gamma$ (see \eqref{eqn:Gamma}) is not unique. Notice that, restricting to odd integers $L\in 2\N+1$, one has the convenient decomposition
\footnote{\label{fn:disjun}The symbol $\bigsqcup$ denotes the disjoint union up to zero-measure sets.} 
\begin{equation}
\label{eqn:QL is L^d Q1}
\FC_{L}= \bigsqcup_{\gamma\in \Gamma \cap \FC_L} \mathrm{T}_\gamma\FC_1.
\end{equation}

We call an operator $A$ acting in $\Hi$ \emph{trace class on compact sets} if  $\chi_K A \chi_K$ is trace class for all compact sets $K \subset \X$
\footnote{This condition is automatically satisfied in the discrete case for any operator $A$, since the range of $\chi_K$ is finite-dimensional.}. 

\begin{definition}[Trace per unit volume]
\label{defn:tuv}
Let $A$ be an operator acting in $\Hi$ such that $ A $ is trace class on compact sets. The \emph{trace per unit volume} of $A$ is defined as
\begin{equation}
\label{eqn:defn tau}
\tau(A):=\lim_{\substack{L\to\infty\\L\in 2\N+1}}\frac{1}{\abs{\FC_L}}\Tr(\chi_L A \chi_L),
\end{equation}
whenever the limit exists.
\end{definition}

Let us denote 
\begin{gather*}
\B_\infty^\tau := \set{\text{bounded \emph{periodic} operators on } \Hi}, \\
\B_1^\tau := \set{\text{$A\in \B_\infty^\tau$ such that } \norm{A}_{1,\tau}:=\tau(\abs{A}) < \infty}.
\end{gather*}
We will refer to operators in $\B_1^\tau$ as the operators of \emph{trace-per-unit-volume class} or \emph{$\tau$-class} for simplicity. 
Moreover, in view of \cite[Proposition~3.17]{BoucletGerminetKleinSchenker05} we have $\B_\infty^\tau \,\cdot\, \B_1^\tau \subset \B_1^\tau$ and  $\B_1^\tau \,\cdot\, \B_\infty^\tau \subset \B_1^\tau$, and 
\begin{equation}
\label{eqn:Binftytau B1tau Binftytau is B1tau}
\norm{AB}_{1,\tau}\leq\norm{A}_{1,\tau}\norm{B}\text{ and }\norm{BA}_{1,\tau}\leq\norm{A}_{1,\tau}\norm{B}\quad \forall \,A\in \B_1^\tau,\: B\in  \B_\infty^\tau\,.
\end{equation}

The following Lemma recalls some useful properties of $\tau$-class operators.
	
\begin{lemma}
\label{lem:tau cont funct1andL}
Let $A\in \B_1^\tau$. Then 
\begin{equation}
\label{eqn:tau cont funct}
\Tr(\abs{\chi_1 A \chi_1})\leq\norm{A}_{1,\tau}
\end{equation}
and
\begin{equation}
\label{eqn:tau cont functL}
\Tr(\abs{\chi_L A \chi_L})<\infty\quad\forall\, L\in 2\N+1.
\end{equation}
In particular, we have that $A$ is trace class on compact sets.
\end{lemma}

\begin{proof}
Inequality~\eqref{eqn:tau cont funct} is proved in \cite[Lemma~3.10]{BoucletGerminetKleinSchenker05} and its proof easily generalizes to obtain also \eqref{eqn:tau cont functL}, as follows. Let $A=U\abs{A}$ be the polar decomposition of $A$. Notice that for any $\gamma,\nu\in\Gamma$
\[
T_{\gamma} \chi_1 T^*_{\gamma} \,A \,T_{\nu} \chi_1 T^*_{\nu}=T_{\gamma} \chi_1 T^*_{\gamma}\, U\abs{A}\, T_{\nu} \chi_1 T^*_{\nu}\text{ is trace class},
\]
since $T_{\gamma} \chi_1 T^*_{\gamma}\, U\abs{A}^{1/2}$ and $\abs{A}^{1/2}\, T_{\nu} \chi_1 T^*_{\nu}$ are Hilbert–Schmidt operators. Thus, 
\[
\chi_L \,A\, \chi_L=\sum_{\gamma,\nu\in \Gamma\cap \FC_L}T_{\gamma} \chi_1 T^*_{\gamma} \,A\, T_{\nu} \chi_1 T^*_{\nu}\text{ is trace class.}
\qedhere\]
\end{proof}
	
The next result allows to compute the trace per unit volume of operators which are periodic and trace class on compact sets.

\begin{proposition} 
\label{prop:charge-tauPeriodic}
\begin{enumerate}[label=(\roman*), ref=(\roman*)]
\item \label{item:per+traceclassoncompact} Let $A$ be periodic and trace class on compact sets\footnote{\label{fn:suffcondtraceclasscptsets} The condition that $A$ is trace class on compact sets is satisfied whenever $A$ is in $\B_1^\tau$, as proved in Lemma \ref{lem:tau cont funct1andL}. 
}. Then $\tau(A)$ is well-defined and 
\begin{equation}
\label{eqn:tauPeriodic1}
\tau(A) =\frac{1}{\abs{\FC_1}} \Tr(\chi_1 A \chi_1).
\end{equation}
\item \label{item:per+traceclassfibr} Let $A$ be a periodic and bounded operator acting on $\Hi$.
Denoting by
\[ \UZ \, A \, \UZ^{-1} = \int_{\mathbb{R}^d}^{\oplus}\di k\, A(k)   \]
its Bloch--Floquet--Zak decomposition, assume that $A(k)$ is trace class and that $\Tr_{\Hf}(\abs{A(k)})<C$ for all $k\in \mathbb{B}^d$. Then
\begin{equation}
\label{eqn:tauPeriodic2}
\Tr(\chi_1 A \chi_1) = \frac{1}{|\mathbb{B}^d|} \int_{\mathbb{B}^d}  \di k\,\Tr_{\Hf}(A(k)). 
\end{equation}
\end{enumerate} 
\end{proposition}
\begin{proof}
{\it \ref{item:per+traceclassoncompact}} In view of the decomposition \eqref{eqn:QL is L^d Q1} and the hypotheses on $A$, one has 
\begin{align*}
\Tr(\chi_L A \chi_L)=\sum_{\gamma\in \Gamma\cap \FC_L}\Tr(T_{\gamma} \chi_1 T^*_{\gamma} A T_{\gamma} \chi_1 T^*_{\gamma})=\sum_{\gamma\in \Gamma\cap \FC_L}\Tr(\chi_1 A \chi_1).
\end{align*}
Since $\abs{\FC_L}=L^d\abs{\FC_1}=\mathrm{card}(\Gamma\cap \FC_L)\abs{\FC_1}$ for every $L\in 2\N+1$, one obtains
\[
\lim_{\substack{L\to\infty\\L\in 2\N+1}}\frac{1}{\abs{\FC_L}}\Tr(\chi_L A \chi_L)=\frac{1}{\abs{\FC_1}}\Tr(\chi_1 A \chi_1).
\]
{\it \ref{item:per+traceclassfibr}} This is proved \eg in \cite[Lemma~3]{PanatiSparberTeufel09}. 
\end{proof}

In the following result, we introduce a class of operators which are not necessarily in $\B_1^\tau$, but have finite trace per unit volume.

\begin{proposition} 
\label{prop:AXi}
Let $A$ be periodic and trace class on compact sets$^{\ref{fn:suffcondtraceclasscptsets}}$. Then 
\begin{enumerate}[label=(\roman*), ref=(\roman*)]
\item \label{it:AXi exh dep} the operator $X_jA$ for $j \in \set{1,\ldots,d}$ has finite trace per unit volume and 
\begin{equation} \label{tau(XA)}
\tau(X_j A ) = \frac{1}{\abs{\FC_1}}\Tr\left( \chi_1 X_j A \chi_1 \right).
\end{equation}
\item \label{it:AXi exh indep} If, in addition  $\tau(A)=0$, then $\tau(X_j A)$ does not depend on the exhaustion\footnote{\label{fn:exh sym} Notice that this particular choice of the exhaustion $\FC_L\nearrow\X$ is such that $\FC_L \cap \Gamma$ is symmetric with respect to the involution $x \mapsto - x$.} $\FC_L\nearrow\X$ set in Definition~\ref{defn:tuv} and on the choice of the origin, in the sense that 
\[
\tau((X_j +\alpha) A )=\tau(X_j  A )\quad\forall\,\alpha\in\R.
\]
\end{enumerate}
\end{proposition}

\noindent As already pointed out, in general it might happen that $X_jA$ is not densely defined, or even that it is trivially defined only on the zero vector of the Hilbert space.

\begin{proof}
{\it \ref{it:AXi exh dep}} Since $\chi_L X_j \chi_L$ is bounded for every $L\in 2\N+1$ and $A$ is trace class on compact sets by hypothesis, we have that $\chi_L X_j  A  \chi_L=\chi_L X_j\chi_L A\chi_L$ is trace class. Therefore, in view of the decomposition \eqref{eqn:QL is L^d Q1}, one has that
\[
\Tr\(\chi_L X_j A \chi_L\)=\sum_{\gamma\in \Gamma\cap \FC_L}\Tr\(T_{\gamma} \chi_1 T^*_{\gamma} X_j A  T_{\gamma} \chi_1 T^*_{\gamma}\)=\sum_{\gamma\in \Gamma\cap \FC_L}\Tr\( \chi_1 T^*_{\gamma} X_j A  T_{\gamma} \chi_1\).
\]
Using that $A$ is periodic, that $[T_{\gamma},X_j]=-\gamma_j T_{\gamma}$, and the result from Proposition~\ref{prop:charge-tauPeriodic}\ref{item:per+traceclassoncompact}, we obtain that
\[
\Tr\( \chi_1 T^*_{\gamma}X_j A  T_{\gamma} \chi_1\)=\Tr\( \chi_1 (X_j +\gamma_j) A \chi_1\)=\Tr\( \chi_1  X_j A  \chi_1\)+\gamma_j\abs{\FC_1}\tau( A).
\]
Consequently, we get that
\begin{equation}
\label{eqn:AXj}
\frac{1}{L^d\abs{\FC_1}}\Tr(\chi_L X_j A \chi_L)=\frac{1}{\abs{\FC_1}}\Tr( \chi_1  X_j A  \chi_1)+\frac{\tau(A)}{L^d} \left(\sum_{\gamma\in \Gamma\cap \FC_L}\gamma_j\right).
\end{equation}
Since both $\gamma$ and $-\gamma$ are in $\Gamma\cap \FC_L$ for all $L\in 2\N+1$, the sum in brackets on the right-hand side of the above vanishes, and the thesis follows immediately.

{\it \ref{it:AXi exh indep}} The statement follows from \eqref{eqn:AXj} and the hypothesis $\tau(A)=0.$ 

\end{proof}

A property which will be fundamental for all the following analysis is the \emph{conditional cyclicity} of the trace per unit volume. We state it in the following Lemma, whose proof can be found in \cite[Lemma 3.22]{BoucletGerminetKleinSchenker05}.

\begin{lemma}[Conditional cyclicity of the trace per unit volume]
\label{prop:cycl of tau}
If $A\in\B_1^\tau$ and $B\in\B_\infty^\tau$, then $\tau(AB)=\tau(BA).$
\end{lemma}

\bigskip

The trace per unit volume is defined in \eqref{eqn:defn tau} through a specific choice of the cell $\FC_L$, which in turn depends via \eqref{eqn:defn FCL} on the choice of a particular linear basis $\{a_1,\ldots,a_d\}$ for $\Gamma$. The term \emph{Unit Cell Consistency} refers to the requirement that physically relevant quantities are independent of the latter choice. Precisely, one considers a different linear basis $\set{\tilde{a}_1, \ldots, \tilde{a}_d}$ for $\Gamma$ and the corresponding cell, defined by
\begin{equation}
\label{eqn:defn tildeFCL}
\widetilde{\FC}_L:=\left\lbrace x\in \X : x = \sum_{j=1}^d \alpha_j \, \tilde{a}_j \text{ with }  |\alpha_j|\leq L/2 \; \forall\: j \in \set{1,\ldots,d}\right\rbrace \,,
\end{equation} 
and sets $\widetilde{\chi}_L:=\chi_{\widetilde{\FC}_L}$. Denoting by $\tau(\,\cdot\,)$ and $\widetilde{\tau}(\,\cdot\,)$, respectively, the trace per unit volume induced by the choice of the primitive cells $\FC_1$ and $\widetilde{\FC}_1$, we prove in Proposition~\ref{cor:tauindprcell}\ref{it: tau indep for Aperiodic} that for any periodic operator $A$, which is trace class on compact sets, one has that
\[
\tau(A)=\widetilde{\tau}(A).
\]
When a contribution to the transport coefficient is in the form  $\tau(X_i A)$ for a periodic operator $A$, as in formula \eqref{Intro:sigma_rot}, a more careful analysis is needed, as discussed at the end 
of Section \ref{ssec:linresponsecoeff} and in Appendix \ref{app:Persistent}.


\section{The unperturbed model} \label{sec:H0}

Our goal is to study the linear response of a crystalline system to the application of an external electric field of small intensity. Before considering the perturbed system, we state our assumptions on the unperturbed one.

\begin{assumption} \label{assum:1} We assume the following:  
\begin{enumerate}[label=$(\mathrm{H}_{\arabic*})$,ref=$(\mathrm{H}_{\arabic*})$]
\item \label{item:smooth_H}
the Hamiltonian $H_0$ of the unperturbed system is a self-adjoint periodic operator on $\Hi$, bounded from below, such that in Bloch--Floquet--Zak representation its fibration 
\[
H_0: \R^d\to \mathcal{L}( \Df,\Hf)\,,\quad
k\mapsto H_0(k) \,,
\]
is a smooth equivariant map taking values in the self-adjoint operators with dense domain $\Df\subset\Hf$, such that 
$\varrho(\gamma^*)\colon\Df\to \Df$ for every $\gamma^*\in\Gamma^*$ (compare \eqref{eqn:rho representation}). 
Here $\mathcal{L}( \Df,\Hf)$ denotes the space of bounded linear operators from $\Df$, equipped with the graph norm of  $H_0(0)$ denoted by $\norm{\,\cdot\,}_{\Df}$
\footnote{From now on $\Df$ is understood to be equipped with the norm $\norm{\,\cdot\,}_{\Df}$.
}
, to $\Hf\simeq L^2(\FC_1) \otimes \C^N$;
\item \label{item:gap} let $\mu\in\R$ (Fermi energy) be in a spectral gap%
\footnote{In the following, when we refer to ``the'' spectral gap of $H_0$, we will refer to this specific gap.} %
of $H_0$. We denote by $\Pi_0 = \chi_{(-\infty, \mu)}(H_0)$ the corresponding spectral projector (Fermi projector). 
We assume that its fibration $k\mapsto\Pi_0(k)$ takes values in the finite-rank projections on $\Hf$\footnote{From the smoothness assumption \ref{item:smooth_H}, it follows that $\mathrm{Rank}(\Pi_0(k)) =m\in\N \cup \set{+ \infty}$ 
 is independent of $k$. Therefore, in view of the fact that $\Pi_0$ is an orthogonal projection, $m<+ \infty$ is equivalent to the assumption  $\Pi_0\in\B_1^\tau$.}. 
\end{enumerate} \hfill $\diamond$
\end{assumption}


We shortly discuss sufficient conditions implying that Assumption \ref{item:smooth_H} holds true.  
As far as discrete models are concerned, $\Hf$ is finite dimensional and $H_0(k)$ are self-adjoint matrices. 
The smoothness of the map  $k \mapsto H_0(k)$ follows from the fact that the hopping amplitudes in the model 
$\set{t_{\gamma}}_{\gamma \in \Gamma}$ decay sufficiently fast as $|\gamma| \to \infty$. 
In all the most  popular discrete  models of topological insulators \cite{Haldane88, KaneMele2005a} the hopping amplitudes have finite range, namely $t_\gamma =0$ if $|\gamma| > R$ for some $R$, hence assumption \ref{item:smooth_H} is automatically satisfied.

As for the continuum case $\X = \R^d$, we first consider a Bloch-Landau operator in the form
\begin{equation} \label{Bloch-Landau}
H_0 = \frac{1}{2} \left( - \iu \nabla - \frac{1}{c} A \right)^2 + V_\Gamma,
\end{equation}
acting in $L^2(\R^d)$, where $A$ and $V_\Gamma$ are the magnetic and electrostatic potentials, respectively
(the charge {\tiny Q} of the particle is reabsorbed in $A$ and in $V$). 
For the sake of simplicity, we consider only $d \le 3$ and we ignore the ``spin space'' $\C^N$, but similar results
holds true if for example $V_\Gamma$ is matrix-valued and acts non-trivially on these degrees of freedom. 
With the help of Kato's theory \cite{Kato66}, and arguing as in \cite{MonacoPanatiPisanteTeufel16} on the basis 
\cite{BirmanSuslina}, it is not difficult to 
prove that, if $A = A_\Gamma$ is $\Gamma$-periodic, and $\mathcal{C}_1$ denotes the fundamental cell of the lattice, then 
for the validity of  \ref{item:smooth_H}  it is sufficient to assume either of the following two sets of hypotheses:
\begin{enumerate}[label=(\Alph*), ref=(\Alph*)]
 \item \label{item:typeA} $A \in L^\infty(\mathcal{C}_1;\R^2)$ when $d=2$ or $A \in L^4(\mathcal{C}_1;\R^3)$ when $d=3$, and $\mathrm{div}\, A, V_\Gamma \in L^2\sub{loc}(\R^d)$ when $d \in \set{2,3}$;
 \item \label{item:typeB} $A \in L^r(\mathcal{C}_1;\R^2)$ with $r>2$ and $V_\Gamma \in L^p(\mathcal{C}_1)$ with $p>1$ when $d=2$, or $A \in L^3(\mathcal{C}_1;\R^3)$ and $V_\Gamma \in L^{3/2}(\mathcal{C}_1)$ when $d=3$.
\end{enumerate}
If instead $A = A_b$ is a linear potential inducing a constant uniform magnetic field, 
it is enough to assume that $V_\Gamma$ is infinitesimally form bounded with respect to $- \Delta$ on $\Hf$, 
and that the magnetic flux per unit cell is a \emph{rational multiple} of the magnetic flux quantum.  
As a further example, we consider the Hamiltonian
\[ H_0 = \frac{1}{2} \V{p}^2 \otimes \Id_{\C^2} + (E_1 p_2 -E_2p_1)\otimes s_z + E_3 (p_1\otimes s_y-p_2\otimes s_x), \]
acting on $L^2(\R^2) \otimes \C^2$. In the above, $\V{p} \equiv (p_1, p_2) = - \iu \nabla$ denotes the momentum operator, $\V{E} \equiv (E_1, E_2, E_3)$ is a constant vector (which we interpret as a constant electric field), and $\V{s} \equiv (s_x, s_y, s_z)$ denotes the vector of spin matrices (half of the Pauli matrices). Thus, the first term in $H_0$ represents the kinetic energy, the second one a spin-orbit coupling, and the third one is a Rashba term: this Hamiltonian represents a continuous analogue of the Kane--Mele Hamiltonian proposed in~\cite{KaneMele2005a}. One can argue that the above operator can be fibered via Bloch--Floquet--Zak transform leading to a family of fiber Hamiltonians $H_0(k)$ as in Assumption~\ref{item:smooth_H}.
Moreover, since 
\[ [H_0,S] = [H_0\su{R},S] = E_3 \left(p_1 \otimes [s_y, s_z] - p_2 \otimes [s_x,s_z] \right) = \iu E_3 \left(p_1 \otimes s_x + p_2 \otimes s_y \right), \]
and since the momentum operator is relatively bounded with respect to the Laplacian, one can see that $[H_0,S]$ is relatively bounded with respect to $H_0$, and hence the assumptions on $S$ (compare Definition~\ref{def:$S$-current} and Remark~\ref{rem:about def of S}) are satisfied as well.
\bigskip
 

The following spaces of operators and functions turn out to be useful for our analysis.

\begin{definition}
Let $\Hi_1,\Hi_2\in \{ \Df \, ,\,\Hf \}$. We denote by $\mathcal{L}(\Hi_1,\Hi_2)$ the space of bounded linear operators from $\Hi_1$ to $\Hi_2$ and by $\mathcal{L}(\Hi_1):=\mathcal{L}(\Hi_1,\Hi_1)$.
We define
\[
\PO(\Hi_1,\Hi_2) := \{   \mbox{  $\Gamma$-periodic $A$ with smooth fibration $\R^d\to \mathcal{L}(\Hi_1,\Hi_2)$, $k\mapsto A(k)$  } \}
\]
equipped with the norm $\norm{A}_{\PO(\Hi_1,\Hi_2)}:=\max_{k\in \mathbb{B}^d}\norm{A(k)}_{\mathcal{L}(\Hi_1,\Hi_2)}$. We also set $\PO(\Hi_1) :=\PO(\Hi_1,\Hi_1)$.
\end{definition}

Since the Fr\'echet derivative follows the usual rules of the differential calculus, we have that $\PO(\Hi_1,\Hi_2)$ is a linear space, $\PO(\Hf)$, $\PO(\Df)$ and $\PO( \Hf,\Df )$ are normed algebras, and \eg
for $A\in\PO(\Hf,\Df)$ and $B\in\PO(\Hf)$ we have
\[
AB\in\PO(\Hf,\Df)\text{ with }\norm{AB}_{\PO(\Hf,\Df)}\leq \norm{A}_{\PO( \Hf,\Df)}\norm{B}_{\PO(\Hf)}\,.
\]

\begin{definition} \label{dfn:Cinfty}
Let $\Hi_1\in \{ \Df \, ,\,\Hf \}$. We set
\[
C^\infty_\varrho(\R^d,\Hi_1):=\{ \varphi\in\Hi_\varrho\text{ such that }\varphi\colon\R^d\to\Hi_1\text{ is smooth}\}.
\]
\end{definition}
Notice that $C^\infty_\varrho(\R^d,\Hf)\supset C^\infty_\varrho(\R^d,\Df)$, $C^\infty_\varrho(\R^d,\Hf)$ is dense   in $\Hi_\varrho$ with  respect to $\norm{\,\cdot\,}_{\Hi_\varrho}$, and $C^\infty_\varrho(\R^d,\Df)$ is dense in $L^2_{\varrho}(\R^d,\Df)$ with respect to the norm on $\Hi_\varrho$ induced by the graph norm $\norm{\,\cdot\,}_{\Df}$.

Since we are interested in computing $[A,X_j]$ where $A$ is in one of the above spaces of operators, the following  Proposition will be relevant. It states their invariance under the derivation $\overline{[\,\cdot\,,X_j]}$, where the overline denotes the operator closure.

\begin{proposition}
\label{prop:derivation of P spaces}
Let $\Hi_1,\Hi_2\in \{ \Df \, ,\,\Hf \}$, and $A\in\PO(\Hi_1,\Hi_2)$. Then $\overline{[A,X_j]}$ is in $\PO(\Hi_1,\Hi_2)$, and  
\[
\overline{[A,X_j]\Big|_ {\UZ^{-1}{C^\infty_\varrho(\R^d,\Hi_1)}}}(k)=-\iu\partial_{k_j}A(k)\text{  in  $\mathcal{L}(\Hi_1,\Hi_2)$}.
\]
\end{proposition}
\begin{proof}
Notice that 
\begin{equation}
\label{eqn:Xj=ipartialk_j}
\UZ X_j \UZ^{-1}\Big|_{C^\infty_\varrho(\R^d,\Hi_1)}= \iu \partial_{k_j}\Big|_{C^\infty_\varrho(\R^d,\Hi_1)},
\end{equation}
thus for every $\varphi\in C^\infty_\varrho(\R^d,\Hi_1)$ one has that
\begin{equation}
\label{eqn:[A,X_j]PHf}
\UZ \, [A,X_j] \, \UZ^{-1}\varphi =\left[ \int_{\mathbb{R}^d}^{\oplus}\di k\, A(k), \iu \partial_{k_j}\right]\varphi=-\iu \int_{\mathbb{R}^d}^{\oplus}\di k\,  \partial_{k_j}A(k) \varphi.
\end{equation}
Since $\UZ \, A \, \UZ^{-1} \, C^\infty_\varrho(\R^d,\Hi_1)\subset C^\infty_\varrho(\R^d,\Hi_2)$ (see \cite[III-\S 3.1, Problem~(3.11)]{Kato66}), the commutator appearing on the right-hand side of the first equality is densely defined on $C^\infty_\varrho(\R^d,\Hi_1)$ and so by unitary conjugation the commutator on the left-hand side is densely defined as well. Thus, Lemma~\ref{lemma:derivata} implies that $[A,X_j]$ acting on ${\UZ^{-1}C^\infty_\varrho(\R^d,\Hi_1)}$ is periodic.

Observe that
\begin{equation}
\label{eqn:partialk_jAPHf}
\norm{\partial_{k_j}A(k)f(k,\cdot)}_{\Hi_2}\leq  \norm{\partial_{k_j}A }_{\PO(\Hi_1,\Hi_2)} \norm{f(k,\cdot)}_{\Hi_1}
\end{equation}
for every $f(k,\cdot) \in  \Hi_1$. By \eqref{eqn:[A,X_j]PHf} and \eqref{eqn:partialk_jAPHf}, one obtains
\[
\norm{   [A,X_j]  (k)\varphi(k,\cdot)}_{\Hi_2}\leq \norm{\partial_{k_j}A }_{\PO(\Hi_1,\Hi_2)}\norm{\varphi(k,\cdot)}_{\Hi_1}
\]
for all $\varphi \in C^\infty_\varrho(\R^d,\Hi_1)$. Therefore, as $\Hi_2$ is a Banach space, the {extension principle} implies the thesis.
\end{proof}

\begin{lemma}
\label{lem:resolvent-projectionPHfDf}
Under Assumption~\ref{assum:1} we have that $(H_0 - z \Id)^{-1}\in \PO( \Hf,\Df )$ for every $z\in \rho(H_0)$, and that $\Pi_0 \in \PO( \Hf,\Df )$.
\end{lemma}
\begin{proof}
The first claim is evident because of $(H_0 - z \Id)^{-1}(k) = (H_0(k) - z \Id)^{-1}$. Since $\Df$ is a Banach space, the second one follows from Riesz's formula
\begin{equation} \label{eqn:Riesz}
\Pi_0(k) = \frac{\iu}{2\pi} \oint_C (H_0(k) - z \Id)^{-1} \, \di z\,,
\end{equation}
where $C$ is a positively-oriented complex contour intersecting the real axis at the Fermi energy (so, in the gap) and below the bottom of the spectrum of $H_0$.
\end{proof}

\begin{corollary}
\label{cor:der Pi0 and H0}
Under Assumption~\ref{assum:1} we have that
\begin{enumerate}[label=(\roman*), ref=(\roman*)]

\item \label{item:[Pi0,X_j]} $\overline{[\Pi_0,X_j]\Big|_ {\UZ^{-1}{C^\infty_\varrho(\R^d,\Hf)}}}\in \PO(\Hf,\Df )$,

\item \label{item:[[Pi0,X_j],X_i]} $\overline{\big[ [\Pi_0,X_j] , X_i \big]\Big|_ {\UZ^{-1}{C^\infty_\varrho(\R^d,\Hf)}}}\in \PO(\Hf,\Df )$,

\item \label{item:[H0,X_j]} $\overline{[H_0,X_j]\Big|_{\UZ^{-1}{C^\infty_\varrho(\R^d,\Df)} }}\in \PO( \Df,\Hf )$.
\end{enumerate}
\end{corollary}

\begin{proof}
{\it \ref{item:[Pi0,X_j]}} By Lemma~\ref{lem:resolvent-projectionPHfDf}, one has that $\Pi_0 \in \PO( \Hf,\Df )$. Proposition~\ref{prop:derivation of P spaces} implies the statement.

{\it \ref{item:[[Pi0,X_j],X_i]}} In view of Lemma~\ref{lem:resolvent-projectionPHfDf}, one has that $\Pi_0 \in \PO( \Hf,\Df )$. Using an argument similar to the one presented in the proof of Proposition~\ref{prop:derivation of P spaces} one deduces the thesis.

{\it \ref{item:[H0,X_j]}} Since by hypothesis $H_0\in \PO( \Df,\Hf )$, Proposition~\ref{prop:derivation of P spaces} concludes the proof.
\end{proof}

For the sake of readability, we introduce the concise notation 
\begin{equation} \label{eqn:closure der of Pi0 and H0}
\begin{aligned}
\overline{[\Pi_0,X_j]}&:=\overline{[\Pi_0,X_j]\Big|_ {\UZ^{-1}{C^\infty_\varrho(\R^d,\Hf)}}}\quad\text{,}\quad\overline{\big[ [\Pi_0,X_j], X_i \big]}:=\overline{\big[[\Pi_0,X_j],X_i\big]\Big|_ {\UZ^{-1}{C^\infty_\varrho(\R^d,\Hf)}}} \cr
&\quad\text{and}\quad\overline{[H_0,X_j]}:=\overline{[H_0,X_j]\Big|_{\UZ^{-1}{C^\infty_\varrho(\R^d,\Df)} }}\,.
\end{aligned}
\end{equation}

\goodbreak


\section{Non-equilibrium almost-stationary states} \label{sec:adiabatic}

Now that we have established the model for the unperturbed system, we consider the perturbed Hamiltonian
\begin{equation} \label{eqn:Hamiltonian}
H^\eps := H_0 - \eps X_j\,,
\end{equation}
where $ \eps \in [0,1]$ is the strength of the external electric field pointing in the $j$-direction. 

As discussed in the introduction, we are interested in the linear response of the system to such a perturbation when it starts initially in the zero-temperature equilibrium state $\Pi_0$. While it is clear that the perturbation given by the linear electric potential has the effect of driving the system out of equilibrium, the perturbation is slowly varying and thus acts locally merely as a shift in energy. Hence it is expected that the initial equilibrium state $\Pi_0$ changes continously into a nearby  \emph{non-equilibrium almost stationary state} (NEASS). A detailed discussion and justification of the concepts of NEASS can be found in 
\cite{Teufel19,MarcelliTeufel19}. 

For the following construction of the NEASS in the present setting we only need to know that 
the operator $\Pi^\epsi$, representing the NEASS, is determined uniquely (up to terms of order $\Or(\epsi^{M+1})$) by the following two properties:
\begin{enumerate}[label=(SA$_{\arabic*}$), ref=(SA$_{\arabic*}$)]
\item \label{item:SA1} $\Pi^\eps = \E^{-\I \epsi \mathcal{S}^\epsi} \, \Pi_0 \, \E^{\I\epsi \mathcal{S}^\epsi}$ for some bounded,  periodic and self-adjoint operator $\mathcal{S}^\epsi$;
\item \label{item:SA3} $\Pi^\eps$ almost-commutes with the Hamiltonian $H^\eps$, namely 
 $[H^\eps,\Pi^\eps] = \Or(\eps^{M+1})$.
\end{enumerate}
Here $\Or(\eps^{M+1})$ is understood in the sense of the operator norm.

\begin{proposition} \label{prop:Pi1}
Consider the Hamiltonian $H^\eps = H_0 - \epsi X_j$ 
with $H_0$ satisfying Assumption~\ref{assum:1}.
\begin{enumerate}[label=(\roman*), ref=(\roman*)]
\item \label{item:propSA1} Let $\mathcal{S}:=\I\, \mathcal{I}([\overline{[X_j,\Pi_0]},\Pi_0])$, where 
\begin{equation} \label{eqn:I(A)}
\mathcal{I}(A) := \frac{\iu}{2 \pi} \oint_C \di z \, (H_0 - z \Id)^{-1} \, [A,\Pi_0] \, (H_0 - z \Id)^{-1},
\end{equation}
with $C$ a positively-oriented contour in the complex energy plane enclosing the part of the spectrum of $H_0$ below the gap. Then $\mathcal{S}$ {is in $\PO(\Hf,\Df)$ and} is self-adjoint.
\item \label{item:propSA3} Let
\[
\Pi^\eps := \E^{-\I \epsi \mathcal{S}} \,\Pi_0\, \E^{\I\epsi \mathcal{S}}\,.
\]
Then
\[
\Pi^\eps = \Pi_0 + \eps \Pi_1 + \eps^2 \Pi_r^\eps
\]
where both
\[
 \Pi_1    	= \mathcal{I}\( \overline{[  X_j,\Pi_0]} \)
 \]
and $ \Pi_r^\eps$ 
are in $\PO(\Hf,\Df)$, and the map ${[0,1]}\ni \eps \mapsto \Pi_r^\eps \in \PO(\Hf,\Df)$ is bounded.
Moreover,
  \begin{align*}
\overline{[H^\eps, \Pi^\eps]}  = \eps^{2} R^\eps
\end{align*}
where $R^\eps$ is in $\PO(\Hf)$ {and the map ${[0,1]} \ni \eps \mapsto R^\eps \in \PO(\Hf)$ is bounded}.

\end{enumerate}
\end{proposition}

We postpone the proof of the above Proposition to Section~\ref{sec:ProofNEASS}. It is already clear from the statement that the map $\mathcal{I}(\,\cdot\,)$ plays a crucial role: its properties are summarized in Section~\ref{sec:ProofLiouvillian}, where we recall in particular the well-known fact from perturbation theory that $\mathcal{I}(A)$ is the unique solution to the equation $[H_0,\mathcal{I}(A)]=A$ whenever $A$ is \emph{off-diagonal} in the orthogonal decomposition induced by $\Pi_0$, namely 
\[ A = A\su{OD} := \Pi_0 \, A \, (\Id-\Pi_0) + (\Id-\Pi_0) \, A \, \Pi_0. \]


\section{Results on the $S$-conductivity} \label{sec:$S$-current}

As stated in the previous sections, we want to investigate quantum $S$-currents induced by the perturbation given by an external electric field, and compute their $S$-conductivities as linear response coefficients. To fix the ideas, the reader can think of the case $S = \Id_\Hi$ (which corresponds physically to the charge current, in appropriate units, \eg in quantum Hall systems) or to $S = \Id_{L^2(\X)} \otimes s_z$, where $s_z = \sigma_z/2$ is half of the third Pauli matrix (which corresponds to the spin current \eg in Quantum Spin Hall systems). 

\begin{definition}[$S$-current and $S$-conductivity] \label{def:$S$-current}
Let $S = \Id_{L^2(\X)} \otimes s$ be a self-adjoint operator on $\Hi = L^2(\X) \otimes \C^N$.
Furthermore, assume that $S$ is periodic%
\footnote{Notice that this assumption is not automatically satisfied since the (magnetic) translation operators (see \eqref{eq:defn transl}) may act non-trivially on the factor $\C^N$. Obviously, for either the \emph{standard magnetic translations} \cite{Zak64} or translation operators with a trivial action on the factor $\C^N$ the periodicity of $S$ is ensured.} %
and its fibration $\Id_{L^2(\FC_1)} \otimes s$ is in $\mathcal{L}( \Df)$.

The \emph{conventional} and the \emph{proper $S$-current} operator are defined respectively as
\begin{equation*}
\begin{aligned}
J\subm{conv}{i}^S&:=\half \( \iu \,\overline{[H_0, X_i]} \, S + \iu \, S \, \overline{[H_0, X_i]} \)\\
J\subm{prop}{i}^S&:=  \iu \, \overline{[H_0,  X_i]}\,S+ \iu\, X_i\, [H_0,  S]
\end{aligned}
\end{equation*}
where $H_0$ satisfies Assumption~\ref{assum:1}. The \emph{conventional} and \emph{proper} $S$-conductivity are defined respectively as
\begin{equation}
\label{eqn:defn sigma}
\re\tau(J\subm{conv/prop}{i}^S\,\Pi^\eps)-  \re \tau(J\subm{conv/prop}{i}^S\,\Pi_0)=:\eps \,\, \sigma\subm{conv/prop}{ij}^S+o(\eps).
\end{equation} \hfill $\diamond$
\end{definition}

\begin{remark}
\label{rem:about def of S} 
Some comments are in order here.
\begin{enumerate}[label=(\roman*), ref=(\roman*)]
\item \label{rem:about def of Jieps} 
By \eqref{eqn:closure der of Pi0 and H0} and $S\in\PO(\Df)$, it is evident that 
\[ J\subm{prop}{i}^S=\iu [H_0, S X_i] \quad  \text{ on }\quad \UZ^{-1}C^\infty_\varrho(\R^d,\Df) \]
and thus $J\subm{prop}{i}^S$ is densely defined as the operator on the right-hand side is so. 
Similarly, we have that
\[
J\subm{conv}{i}^S=\half \( \iu {[H_0, X_i]} \, S + \iu  S \, {[H_0, X_i]} \) \quad  \text{ on } \quad \UZ^{-1}C^\infty_\varrho(\R^d,\Df).
\]
\item \label{rem:about def of S2}
The hypothesis $\Id_{L^2(\FC_1)} \otimes s \in\mathcal{L}( \Df)$ and Assumption~\ref{assum:1} ensure the condition of relative boundedness of $[H_0,S]$ with to respect to $H_0$. 

\end{enumerate}\hfill $\diamond$
\end{remark}

\noindent Since $\Pi^\eps=\Pi_0 + \eps \Pi_1 + \eps^2 \Pi_r^\eps$ by Proposition~\ref{prop:Pi1}\ref{item:propSA3}, we have that 
\begin{equation}
\label{eqn:sigmaprop}
\re\tau(J\subm{prop}{i}^S\,\Pi^\eps)-\re\tau(J\subm{prop}{i}^S\,\Pi_0)=\eps \re\tau(J\subm{prop}{i}^S\,\Pi_1)+\eps^2\re\tau(J\subm{prop}{i}^S\,\Pi_r^\eps).
\end{equation}
In order to prove that $\re\tau(J\subm{prop}{i}^S\,\Pi_1)=\sigma\subm{prop}{ij}^S$ according to \eqref{eqn:defn sigma}, it suffices to show that all the traces per unit volume above are well-defined and finite, and that the term carrying a prefactor $\eps^2$ is uniformly bounded in $\eps$.
While the control of the remainder term will be done in Section~\ref{sec:Proofsigma}, we focus now on the linear response coefficient, namely $\re\tau(J\subm{prop}{i}^S\,\Pi_1)$.

\goodbreak

\subsection{The linear response coefficient}
\label{ssec:linresponsecoeff} 
In order to compute the linear response coefficient, we employ directly Definition~\ref{defn:tuv} for $\tau(J\subm{prop}{i}^S \,\Pi_1)$, and start by localizing this operator on the cell $\FC_L$, defined in~\eqref{eqn:defn FCL}, through the projection $\chi_L$ which multiplies by the characteristic function of $\FC_L$. It is convenient to notice at this point the following

\begin{remark} \label{rmk:SmoothBF}
The range of  $\UZ \chi_L$ is contained in $C^\infty_\varrho(\R^d, \Hf)$ (compare Definition~\ref{dfn:Cinfty}) for every $L>0$. Indeed, for all $f \in \Hi = L^2(\X) \otimes \C^N$ and all ${r} \in \N$, the function $\langle X \rangle^{{r}} \chi_L f$ is still in $\Hi$, where $(\langle X \rangle \psi)(x) := (1 + |x|^2)^{1/2} \psi(x)$ for $\psi\in\Do(\langle X \rangle)$. 
By standard Bloch--Floquet theory \cite[Appendix A]{MonacoPanatiPisanteTeufel16}, 
this is equivalent to requiring that $\UZ (\chi_L f)$ is in the space $H^{{r}}_\varrho(\R^d,\Hf)$ of $\varrho$-covariant maps $\varphi \colon \R^d \to \Hf$ with Sobolev regularity ${r}$: it is a classical result that the intersection of all these Sobolev spaces is contained in $C^\infty_\varrho(\R^d,\Hf)$.
\hfill $\diamond$
\end{remark}

Notice that, since $\Pi_1 \in \mathcal{P}(\Hf,\Df)$, it maps $\UZ^{-1} C^\infty_\varrho(\R^d,\Hf)$ to $\UZ^{-1} C^\infty_\varrho(\R^d,\Df)$. Thus, in view of Remark~\ref{rem:about def of S}\ref{rem:about def of Jieps}, we have that
\begin{equation}
\label{eqn:Ji0Pi1chiL}
J\subm{prop}{i}^S\, \Pi_1 \chi_L = \iu [H_0, S X_i] \Pi_1  \chi_L
\end{equation}
and this allows for the following simple manipulations. Using the Leibniz rule, we obtain the following chain of equalities on $\Ran(\chi_L)$:
\begin{equation}
\label{eq:T+R}
\iu [H_0,S X_i] \Pi_1=\iu [H_0,X_i S ] \Pi_1= \iu  [H_0,X_i] S\Pi_1+\iu X_i [H_0, S ]  \Pi_1=\mathsf{O}+X_i\mathsf{R},
\end{equation}
where the operators $\mathsf{O}$ and $\mathsf{R}$ have been defined in \eqref{eqn:defn T and R}. We call $\mathsf{O}$ the \emph{$S$-orbital} part and $\mathsf{R}$ the \emph{$S$-rotation} part of the operator related to linear response of the proper $S$-current $J\subm{prop}{i}^S$. The latter terminology is due to the fact that, when $S$ is the spin operator and $[H_0,S]\neq 0$, the spin transport is the result of two contributions, that is, the $S$-orbital part coming from the center-of-mass drift and the $S$-rotation part due to the spin non-conservation. Notice that both $\mathsf{O}$ and $\mathsf{R}$ are periodic, instead obviously $X_i\mathsf{R}$ is not periodic and thus a more careful analysis of its trace per unit volume is required. 
We begin by handling the $S$-orbital part $\mathsf{O}$. In view of the defining relation $\Pi_1 = \mathcal{I}(\overline{[X_j,\Pi_0]})$ (compare Propositions~\ref{prop:Pi1} and~\ref{prop:I(A)}), we obtain on $\Ran(\chi_L)$:
\begin{align}
\mathsf{O}&=\iu  [H_0,X_i\su{D}] S \Pi_1+\iu [H_0, X_i\su{OD}] S\Pi_1\cr
&= \mathsf{E}_1+\iu [H_0, X_i\su{OD}S\Pi_1]-\iu X_i\su{OD} [H_0,S\Pi_1] \cr
&=\mathsf{E}_1+\mathsf{E}_2-\iu X_i\su{OD} [H_0,S]\Pi_1 -\iu X_i\su{OD}S [X_j,\Pi_0]\cr
&=\sum_{\ell=1}^3\mathsf{E}_\ell-\iu\big[[X_i,\Pi_0],\Pi_0 S [X_j,\Pi_0] \big]+ \iu\Pi_0\big[[X_i,\Pi_0], S [X_j,\Pi_0] \big] \cr
\label{eqn:K and beyond}
&=\sum_{\ell=1}^4\mathsf{E}_\ell+\mathsf{C},
\end{align}
where we have defined the operators 
\begin{equation}
\label{eqn: defn E and C}
\begin{aligned}
 \mathsf{E}_1&:=\iu  [H_0,X_i\su{D}] S \Pi_1,\quad  \mathsf{E}_2:=\iu [H_0, X_i\su{OD}S\Pi_1],\quad \mathsf{E}_3:=\iu X_i\su{OD} [S,H_0]\Pi_1,\cr
  \mathsf{E}_4&:=\iu\big[[X_i,\Pi_0],\Pi_0 S [\Pi_0,X_j] \big]\quad\text{ and }\quad\mathsf{C}:= \iu\Pi_0\big[[X_i,\Pi_0], S [X_j,\Pi_0] \big].
\end{aligned}
\end{equation}
We call $\mathsf{C}$ the \emph{Chern-like term} and $\mathsf{E}_\ell$ the \emph{$\ell$-th extra or beyond-Chern-like term} for $\ell\in\{1,\dots,4\}$. This terminology is motivated by the fact that, whenever the spin is conserved, for $d=2$, $i=1$ and $j=2$ in quantum (spin) Hall systems the Chern-like term $\mathsf{C}$ corresponds to the (spin) Chern number (see Remark~\ref{rem:Spin chern number}). In general, whenever $[H_0,S]=0$, all extra terms have trace per unit volume zero (see Subsection~\ref{ssec: S conserved}) and obviously the $S$-rotation part vanishes.
 
In the following Proposition, we analyze the trace per unit volume of the operators resulting from the previous algebraic manipulations. 

\begin{proposition}
\label{prop:Extrachanged}
Under Assumption~\ref{assum:1} and hypotheses on $S$ in Definition~\ref{def:$S$-current}, we have that the Chern-like term $\mathsf{C}$, the extra terms $\mathsf{E}_\ell$ for any $\ell\in\{1,\dots,4\}$ and $X_i\mathsf{R}$, defined in \eqref{eqn:defn T and R} and \eqref{eqn: defn E and C} have finite traces per unit volume. Moreover, one has
\begin{equation}
\label{eqn:tau reads only prim cell}
\tau(\mathsf{A})=\frac{1}{\abs{\FC_1}}\Tr\left( \chi_1 \mathsf{A} \chi_1 \right)\text{, for $\mathsf{A}\in\big\{\mathsf{C},\mathsf{E}_\ell, X_i\mathsf{R}\,:\, \ell\in \{1,\dots,4\}, i\in\{1,\dots,d\}\big\}$},
\end{equation}
and
\begin{align}
\label{eqn: tau C}
\tau(\mathsf{C})&=\iu\tau(\Pi_0\big[[X_i,\Pi_0]{S},  [X_j,\Pi_0] \big]),\\
\label{eqn: tau E1 E3}
\tau(\mathsf{E}_1)&=\iu\tau(  [H_0,X_i\su{D}] S\su{OD} \Pi_1),\qquad \tau(\mathsf{E}_3)=\iu \tau(X_i\su{OD} [S,H_0]\Pi_1),\\
\label{eqn: tau E2 E4}
\tau(\mathsf{E}_2)&=0=\tau(\mathsf{E}_4),
\end{align}
where the diagonal and off-diagonal parts of the above operators refer to the orthogonal decomposition induced by the Fermi projection $\Pi_0$.  
\end{proposition}

The proof of the above Proposition is postponed to Section~\ref{sec:ProofKE}.

We are going to prove that trace per unit volume of the operator $X_i\mathsf{R}$ is well-defined and finite.
In view of Proposition~\ref{prop:AXi}\ref{it:AXi exh indep}, it suffices to show that $\tau(\mathsf{R})$ is zero. The latter result is an immediate consequence of the following
\begin{proposition} \label{prop:Btorque}
If $H_0$ satisfies Assumption~\ref{assum:1}, and $B$ is in  $\PO(\Hf)\cap\PO(\Df)$ (in particular, $B$ is a {bounded
periodic} operator) the following holds: 
\begin{enumerate}[label={(\roman*)}, ref={\rm (\roman*)}]
\item \label{item:doublecomm formula} 
$$
\tau\big(\iu [H_0,B] \Pi_1 \big)=\tau \big( \iu\Pi_0\big[ [B,\Pi_0]  ,  [X_j,\Pi_0] \big] \big)
$$ \\[-5mm]
\item \label{item: vanishing of doublecomm formula if A per} If, in addition, $[B,X_j]=0$ then 
\[
\tau(\iu [H_0,B]\Pi_1)=0.
\]
\end{enumerate}
\end{proposition}

The above Proposition, whose proof is deferred to Section~\ref{sec:ProofKE}, immediately implies the following
						
\begin{corollary}
\label{cor:tau R=0}
Under Assumption~\ref{assum:1} and hypotheses on $S$ in Definition~\ref{def:$S$-current}, we have that
\[
\tau(\mathsf{R})=\tau(\iu  [H_0, S ]  \Pi_1)=0.
\]
\end{corollary}

We are now in position to state one of our main results.

\goodbreak

\begin{theorem}[General formula for the $S$-conductivity] 
\label{thm:notKubo}
Let $H^\eps = H_0 - \epsi X_j$ be acting in $L^2(\X) \otimes \C^N$, with $H_0$ and $\Pi_0$ as in Assumption~\ref{assum:1}. Let $\Pi^\eps$ be the NEASS defined in Section~\ref{sec:adiabatic}. Consider the  conventional (resp.\,\,proper) $S$-conductivity 
$\sigma\subm{conv}{ij}$ (resp.\,\,$\sigma\subm{prop}{ij}$) as in Definition~\ref{def:$S$-current}.  
Then
\begin{equation}
\label{eqn:sigmapropersplittinconvandrot}
\sigma\subm{prop}{ij}^S= \sigma\subm{conv}{ij}^S+\sigma\subm{rot}{ij}^S,
\end{equation}
where 
\begin{equation}
\label{eqn:sigmaT}
\begin{aligned}
\sigma\subm{conv}{ij}^S&=\frac{1}{\abs{\FC_1}}\re \Tr \Big( \chi_1\,  \iu \Pi_0\big[[X_i,\Pi_0]S, [X_j,\Pi_0] \big]  \, \chi_1 \Big) \\
& \quad + \frac{1}{\abs{\FC_1}}\re \Tr  \Big(\chi_1\( \iu [H_0,X_i\su{D}] S\su{OD} \Pi_1 + \iu X_i\su{OD} [S,H_0]\Pi_1   \)\chi_1 \Big) 
\end{aligned}
\end{equation}
and the rotation contribution to the proper $S$-conductivity is defined as
\begin{equation}
\label{eqn:sigmaR}
\sigma\subm{rot}{ij}^S=\re \tau(\iu X_i  [H_0, S ]  \Pi_1)=\frac{1}{\abs{\FC_1}}\re\Tr  \Big(\chi_1 \iu X_i  [H_0, S ]  \Pi_1\chi_1 \Big).
\end{equation}
Moreover, the trace per unit volume appearing in \eqref{eqn:sigmaR} does not depend on the particular exhaustion$^{\ref{fn:exh sym}}$ $\FC_L\nearrow\X$ chosen in Definition~\ref{defn:tuv} and on the choice of the origin, in the sense that $\tau(X_i\mathsf{R})=\tau((X_i+\alpha)\mathsf{R})$ for every $\alpha\in\R$.
\end{theorem}

The proof of the above Theorem is postponed to Section \ref{sec:Proof}.

\goodbreak

\begin{remark} Some comments about the above result.
\begin{enumerate}[label=(\roman*), ref=(\roman*)]
\item \label{rem: splittings for sigma}
Notice that one can rewrite the above formula for the proper $S$-conductivity $\sigma\subm{prop}{ij}^S$, summing the two contributions $\sigma\subm{conv}{ij}^S$ and $\sigma\subm{rot}{ij}^S$ as follows:
\begin{align}
\label{eq: per sigma}
\phantom{LLl}\sigma\subm{prop}{ij}^S&=\frac{1}{\abs{\FC_1}}\re \Tr \Big( \chi_1\(\iu \Pi_0\big[[X_i,\Pi_0] S,  [X_j,\Pi_0] \big] +\iu [H_0,X_i\su{D}] S\su{OD} \Pi_1  \) \chi_1\Big)  \\
\label{eq: nonper sigma}
& \quad + \frac{1}{\abs{\FC_1}}\re \Tr  \Big(\chi_1 \iu X_i\su{D}  [H_0, S ]  \Pi_1   \chi_1 \Big).
\end{align}
While \eqref{eqn:sigmapropersplittinconvandrot} emphasizes the splitting between the drift contribution coming from the center-of-mass momentum and the one resulting from the spin rotation, the latter decomposition isolates the contribution coming from a periodic operator, in \eqref{eq: per sigma}, and the one deriving from a non-periodic operator, in \eqref{eq: nonper sigma}.
\item \label{rem:reality of sigma}
The real part is needed in both \eqref{eqn:sigmaT} and \eqref{eqn:sigmaR}, even if on $\Ran(\chi_1)$ one has that $J\subm{prop}{i}^S\, \Pi^\eps  =  {\(J\subm{prop}{i}^S\)}^*\,\Pi^\eps.$
Indeed, notice that
\begin{align*}
&\tau(J\subm{prop}{i}^S\, \Pi^\eps)=\tau(\Pi^\eps \,J\subm{prop}{i}^S\, \Pi^\eps)+\tau({\Pi^\eps}^\perp J\subm{prop}{i}^S\, \Pi^\eps),\cr
&\tau(\Pi^\eps \,J\subm{prop}{i}^S\, \Pi^\eps)\in \R,
\end{align*}
where $\tau({\Pi^\eps}^\perp J\subm{prop}{i}^S\, \Pi_\eps)$ is generally not zero, since the argument of the trace per unit volume is not periodic, due to $[H_0,S]\neq 0$, and thus we are not allowed to use cyclicity. 
On the other hand, if $[H_0,S]=0$ then $\tau({\Pi^\eps}^\perp J\subm{prop}{i}^S\, \Pi^\eps)=0$ and therefore $\tau(J\subm{prop}{i}^S\, \Pi^\eps)$ is automatically real. Moreover, for systems with a fermionic time-reversal symmetry $\Theta$ such that $\Theta S \Theta^{-1} = -S$, the number  $\tau(J\subm{prop}{i}^S\, \Pi_1)$ is real, so the real part is redundant.
\end{enumerate}
\end{remark}

It is worth to investigate how  the contributions to the proper $S$-conductivity, appearing in \eqref{eqn:sigmapropersplittinconvandrot}, behave under a change of primitive cell.
\begin{proposition}[Unit Cell Consistency of the $S$-conductivity]
\label{prop:UCC}
Under the hypotheses of Theorem~\ref{thm:notKubo}, we have that
\begin{enumerate}[label=(\roman*), ref=(\roman*)]
\item \label{it:UCC conv sigma} $\sigma\subm{conv}{ij}^S$ satisfies UCC.
\item \label{it:UCC rot sigma} If, in addition, the model enjoys a discrete rotational symmetry satisfying the hypotheses of Proposition~\ref{prop:van mesopointwise spin torque}, then $\sigma\subm{rot}{ij}^S$ satisfies UCC.
\end{enumerate}
\end{proposition}
\begin{proof}
{\it \ref{it:UCC conv sigma}} Since all operators involved in the trace per unit volume computing $\sigma\subm{conv}{ij}^S$ are periodic, Proposition~\ref{cor:tauindprcell}\ref{it: tau indep for Aperiodic} implies the thesis.
{\it \ref{it:UCC rot sigma}} By applying Proposition~\ref{prop:van mesopointwise spin torque}\ref{item:pointwisemesospintorq} along with Proposition~\ref{cor:tauindprcell}\ref{it: tau dep for X_iA}, the conclusion follows.
 \end{proof}
The next Proposition shows that in some discrete models with discrete rotational symmetry, one has that $\sigma\subm{rot}{ij}^S=0$, and hence the choice between $J\subm{prop}{i}^S$ and $J\subm{conv}{i}^S$ becomes immaterial. Remarkably, the Kane--Mele model is in this class. 
\begin{proposition}[Equality of conventional and proper $S$-conductivity]
\label{prop:Srotsigmazeroundersym}
Let $H_0$ be a \emph{discrete Hamiltonian with finite range hopping amplitudes} and $S$ be as in Definition~\ref{def:$S$-current}. Assume that the model satisfies the hypotheses of Proposition~\ref{prop:van mesopointwise spin torque} and $\mathrm{Rank}{\chi_{P_\gamma}}=1$, where $\set{P_\gamma}_{\gamma\in I}\subset \X$ is the family of subsets defined in \eqref{eq:decomps C_1 and tildeC_1}. Then
\[
\sigma\subm{rot}{ij}^S=0\quad\text{or equivalently}\quad \sigma\subm{prop}{ij}^S=\sigma\subm{conv}{ij}^S.
\]
\end{proposition}
\begin{proof}
By direct computation, since there exists $\lambda_{i,\gamma}\in\R$ such that $X_i\chi_{P_\gamma}=\lambda_{i,\gamma}\chi_{P_\gamma}$, we have that
\begin{align*}
\abs{\FC_1}\sigma\subm{rot}{ij}^S&=\re\Tr  \Big(\chi_1 \iu X_i  [H_0, S ]  \Pi_1\chi_1 \Big)=\sum_{\gamma\in I}\re\Tr(\chi_{P_\gamma}\, X_i  \iu[H_0, S ]  \Pi_1\, \chi_{P_\gamma})\cr
&=\sum_{\gamma\in I}\lambda_{i,\gamma}\re\Tr(\chi_{P_\gamma}\,\iu   [H_0, S ]  \Pi_1\, \chi_{P_\gamma})=0,
\end{align*}
because by Proposition~\ref{prop:van mesopointwise spin torque}\ref{item:pointwisemesospintorq} $\Tr(\chi_{P_\gamma}\,   \iu[H_0, S ]  \Pi_1\, \chi_{P_\gamma})=0$ for every $\gamma\in I$.
\end{proof}

\bigskip
\goodbreak

\subsection{\texorpdfstring{When $S$ is (approximately) conserved}{When S is a conserved quantity}}
\label{ssec: S conserved}
The computation of the linear response coefficient $\sigma_{ij}^S$ simplifies considerably if we assume that $S$ is a conserved quantity of the system, namely that
\begin{equation} \label{Sconserved} [H_0, S] = 0. \end{equation}
Under this assumption, then $[\Pi_0, S]=0$ as well, since $\Pi_0$ is a spectral projection associated to $H_0$,  
and thus $S$ is diagonal in the decomposition induced by $\Pi_0$.

If \eqref{Sconserved} holds, then $J_i^S:=J\subm{prop}{i}^S=J\subm{conv}{i}^S= \iu \overline{[H_0, X_i]} S$ is in $\PO(\Df,\Hf)$ by Corollary~\ref{cor:der Pi0 and H0}\ref{item:[H0,X_j]} and the hypothesis $S\in\PO(\Df)$. Hence, since $\Pi^\eps\in\PO(\Hf,\Df)$ by Proposition~\ref{prop:Pi1}\ref{item:propSA3}, we have that $J_i^S\Pi^\eps\in \PO(\Hf)$ and furthermore applying Proposition~\ref{cor:welldef sigmaijeps}\ref{item:welldef sigmaijeps 1}, we deduce that $J_i^S\Pi^\eps $ is $\tau$-class. Thus, the trace per unit volume of $J_i^S\Pi^\eps $ is well-defined and only the Chern-like term contributes to it. Indeed, by Proposition~\ref{prop:Extrachanged} the extra term $\mathsf{E}_\ell$ does not contribute for $\ell\in\{2,4\}$ and the next Lemma shows that $\tau(\mathsf{E}_1)=0$. Obviously, $\mathsf{E}_3=0=\mathsf{R}$ whenever \eqref{Sconserved} holds.

\begin{lemma} \label{lemma:Extra2}
Under Assumption~\ref{assum:1} and the hypotheses on $S$ in Definition~\ref{def:$S$-current}, assume further that 
$[H_0,S] =0$. Then, 
\[
\tau(\mathsf{E}_1)=0\quad\text{ and }\quad\mathsf{E}_3=0=\mathsf{R}
\]
where $\mathsf{E}_1$, $\mathsf{E}_3$ and $\mathsf{R}$ are defined in \eqref{eqn: defn E and C} and \eqref{eqn:defn T and R}. 
\end{lemma}

The proof of Lemma~\ref{lemma:Extra2} is also deferred to Section~\ref{sec:ProofKE}, but it is easily seen to imply the following

\goodbreak

\begin{theorem}[$S$-conductivity in the $S$ conserved case] \label{thm:Kubo}\ \\
Let $H^\eps = H_0 - \epsi X_j$ be acting in $L^2(\X) \otimes \C^N$, with $H_0$ and $\Pi_0$ as in Assumption~\ref{assum:1}. 
Let $\Pi^\eps$ be the NEASS defined in Section~\ref{sec:adiabatic} and $J_i^S$ be as in Definition~\ref{def:$S$-current}. Assume moreover that  $[H_0,S]=0$. Then the $S$-conductivity is 
\begin{align*}
\sigma_{ij}^S  & =   
\frac{\iu}{\abs{\FC_1}} \Tr \Big(\chi_1 \, S  \Pi_0\big[\, [X_i,\Pi_0] , [X_j,\Pi_0]\, \big]  \,\chi_1  \Big) \\
& = \frac{\iu}{(2\pi)^d} \int_{\mathbb{B}^d} \di k \, \Tr_{\Hf} \left( (\Id \otimes s)\, \Pi_0(k) \left[ \partial_{k_j} \Pi_0(k), \partial_{k_i} \Pi_0(k) \right] \right).
\end{align*}
\end{theorem}
\begin{proof}
In view of Lemma~\ref{lemma:Extra2}, the extra terms $\mathsf{E}_1$ and $\mathsf{E}_3$, and the $S$-rotation part $\sigma\subm{rot}{ij}^S$ do not contribute to the trace per unit volume of $J_i^S\, \Pi_1$. Therefore, using Proposition~\ref{prop:charge-tauPeriodic}\ref{item:per+traceclassfibr} we are going to compute the $k$-space representation of the trace of ${\mathsf{C}} \chi_1 =\iu S \Pi_0 \big[ [X_i,\Pi_0],[X_j,\Pi_0] \big] \chi_1$. 
To this end, it suffices to notice that the fiber operator associated to $\overline{[X_j,\Pi_0]}$ in the Bloch--Floquet--Zak representation is given by $\iu \partial_{k_j} \Pi_0(k)$ (Proposition~\ref{prop:derivation of P spaces}) and that $|\mathcal{C}_1| \, |\mathbb{B}^d| = (2\pi)^d$.
\end{proof}

This Theorem applies in particular to the transverse charge current in quantum Hall systems ($S=\Id_\Hi$), and to the transverse spin current in quantum spin Hall systems ($S=\Id_{L^2(\X)} \otimes s_z$)  
whenever the $z$-component of the spin is conserved. In particular, in the latter case we recover the formula for the spin conductivity proposed in \cite{CulcerYaoNiu05, ShindouImura05}, which was derived assuming that the unperturbed Hamiltonian $H_0$ has an identically degenerate Bloch band, where the degeneracy comes from the spin degrees of freedom. So in that model effectively $(\Pi_0 H_0 \Pi_0)(k) = E_0(k) \Id_{L^2(\X)} \otimes \Id_{\C^2}$, and \eqref{Sconserved} is in particular satisfied after projection to the relevant spectral subspace. Our argument used only \eqref{Sconserved} and no spectral assumption (other than the gap condition) on the Hamiltonian.

\begin{remark}[Spin conductivity and spin-Chern number]
\label{rem:Spin chern number}
Let $S=\Id_{L^2(\X)} \otimes s_z$ with $s_z$ a spin operator for non-integer spin $r$, \ie  with spectrum $\{ -r, -r+1,\ldots, r-1, r\}$, acting on $\C^N$ with $N=2r+1$ (\eg half the third Pauli matrix $\sigma_z$ for $r=\frac12$ and $N=2$). Denote by $s_z = \sum_{\ell=0}^{2r} ( \ell  -r ) p_\ell$ its spectral decomposition.

Then the commutation relation $[\Pi_0,S] = 0$ implies that $\Pi_0$ admits a splitting in the decomposition induced by $S$:
\[ \Pi_0 =  \sum_{\ell=0}^{2r} \Pi_0^{(\ell)} :=   \sum_{\ell=0}^{2r} \Pi_0 (\Id\otimes p_\ell)  \,. \]
The formula for the $S$-conductivity $\sigma_{ij}^S $ in Theorem~\ref{thm:Kubo} simplifies then to
\begin{align}\label{def:SpinChernNumber} 
\sigma_{ij}^S =&\frac{\iu}{\abs{\FC_1}} \sum_{\ell=0}^{2r} (\ell-r) \Tr \Big(\chi_1 \, \Pi_0^{(\ell)} \left[ \big[X_i,\Pi_0^{(\ell)} \big] , \big[X_j,\Pi_0^{(\ell)} \big] \right] \,\chi_1\Big)  =: \frac{1}{(2\pi)^{d-1}} \mbox{ $S$-Chern$(\Pi_0)_{ij}$}.
\end{align}
This {\em spin-Chern number} $S$-Chern$(\Pi_0)_{ij}$,  proposed in \cite{ShengWengShengHaldane2006} and intrinsically defined in \cite{Prodan2009}, is in general a half integer.
\footnote{The normalization we use here agrees with \cite{Prodan2009} and with the most recent physics literature, but differs by a factor $2$ from the original formula in \cite{ShengWengShengHaldane2006}.
}\  
It becomes an integer if the system enjoys time-reversal symmetry.
Even in time-reversal invariant systems it can be different from zero while the {\em Chern number }
\[
\mbox{Chern}(\Pi_0) :=  {\frac{ \iu \, (2\pi)^{d-1}}{\abs{\FC_1}}} \Tr \Big(\chi_1 \, \Pi_0  \left[ \big[X_i,\Pi_0  \big] , \big[X_j,\Pi_0  \big] \right] \,\chi_1\Big)  
\]
necessarily vanishes.
Our approach to spin transport shows then that for $d=2$ the bulk spin Hall conductivity 
(measured in units of $\frac{e}{2\pi} \equiv \frac{1}{2\pi}$) equals the \emph{spin-Chern number}, 
as long as $[H_0, S]=0$. 
On the other hand, when $S=\Id_{L^2(\X)}\otimes \Id_{\C^N}$ and $d=2$ the integral in Theorem~\ref{thm:Kubo} computes,
up to a factor $1/2\pi$, the Chern number  Chern$(\Pi_0)$  of the family of projections $\set{\Pi_0(k)}_{k \in \R^2}$, implying the quantization of the Hall conductivity {measured in units of $\frac{e^2}{h} \equiv \frac{1}{2\pi}$} (see~\cite{Graf07} and references therein).
 
\hfill $\diamond$
\end{remark}

Abstracting from the previous Remark, we consider now any operator in the form $S=\Id_{L^2(\X)} \otimes s$,
with $s$ as in Definition \ref{def:$S$-current}. If $S$ is approximately conserved, \ie if  
\[
\lambda:=   \|[H_0,S]\|_{\PO( \Df,\Hf )} 
\]
is sufficiently small,
then one can still define a spin-Chern number related to $\Pi_0$ \cite{Prodan2009, Schulz-Baldes2013} and the $S$-conductivity is still approximately given by the spin-Chern number. To see this, let $s = \sum_{\ell =1}^k s_\ell \,p_\ell$ be the spectral representation of $s$  (we need no assumptions on the spectrum of $s$ here) and
\[
\widetilde H_0 :=  \sum_{\ell =1}^k (\Id\otimes p_\ell) H_0 (\Id\otimes p_\ell) \quad\mbox{and}\quad V:=   ( H_0 - \widetilde H_0) =   \sum_{\ell_1 \not= \ell_2}  (\Id\otimes p_{\ell_1})H_0 (\Id\otimes p_{\ell_2})\,.
\]
Then 
\[
H_0 =  \widetilde H_0 +  V
\]
where
\[ [\widetilde H_0,S]=0 \]
and 
\[
 \|V\|_{\PO( \Df,\Hf )}  =   \left\| \sum_{\ell_1 \not= \ell_2}   (\Id\otimes p_{\ell_1}) [H_0,  (\Id\otimes p_{\ell_2})]\right\|_{\PO( \Df,\Hf )} \leq \lambda\,C_s  \,,
\]
with a constant $C_s$ that depends only on $S$.

The spin-conserving Hamiltonian $\widetilde H_0$ is $H_0$-bounded with relative bound $\lambda C_s$. For   $\lambda <\frac{1}{C_s}$,    $\widetilde H_0$ is thus self-adjoint on the domain of $H_0$,  and for $\lambda$ small enough, by standard perturbation theory, the Fermi energy $\mu$ lies also in a gap of $\widetilde H_0$. Thus we can define the gapped Fermi  projection $\widetilde\Pi_0 := \chi_{(-\infty,\mu]}(\widetilde H_0)$ of $\widetilde H_0$ and, 
in analogy with Remark~\ref{rem:Spin chern number}, its associated 
spin-Chern number. More precisely, let 
\begin{equation}\label{tildepidef}
\widetilde \Pi_0^{(\ell)} := \widetilde \Pi_0(\Id\otimes p_\ell) \qquad\mbox{and thus}\quad \sum_{\ell=1}^{k} \widetilde \Pi_0^{(\ell)}  = \widetilde \Pi_0\,.
\end{equation}
It is straightforward to see that\footnote{\label{fn:B_1tau prop for pert Pi} By choosing $\lambda$ small enough, one has that 
$\norm{\widetilde \Pi_0(k)- \Pi_0(k)}<1$ and thus the two projections have the same rank for every $k$. Then, 
by the argument in the proof of Lemma \ref{lem:Pi in tau-class}, it follows that $\widetilde \Pi_0\in\B_1^\tau$ and hence $\widetilde \Pi_0^{(\ell)}\in\B_1^\tau$ as well.}  $\widetilde \Pi_0^{(\ell)}\in \PO(  \Hf,\Df )\cap \B_1^\tau$
  and thus the  Chern numbers 
 \[
\mbox{Chern}(\widetilde \Pi_0^{(\ell)} )_{ij}:= 
\frac{ \iu {(2\pi)^{d-1}}}{\abs{\FC_1}} \Tr \Big(\chi_1 \,  \widetilde \Pi_0^{(\ell)}  \left[ \big[X_i,\widetilde \Pi_0^{(\ell)} \big] , \big[X_j,\widetilde \Pi_0^{(\ell)} \big] \right] \,\chi_1\Big)  \in \Z 
\]
are well defined and integer.
 The $S$-Chern number of $ \Pi_0 $ is finally defined as
\begin{equation}\label{SChernDef}
\mbox{$S$-Chern}(  \Pi_0  )_{ij}:=  \sum_{\ell=1}^{k} s_\ell \cdot  \mbox{Chern}(\widetilde \Pi_0^{(\ell)} )_{ij}\,.
\end{equation}

We now show that the $S$-conductivity is given at leading order in $\lambda$ by  $\mbox{$S$-Chern}(  \Pi_0  )$,
a result which coherently complements the robustness of \emph{edge} spin currents proved by Schulz-Baldes 
\cite{Schulz-Baldes2013}. 

To formulate such a perturbative statement precisely, we slightly change perspective and notation and introduce a $\lambda$-dependent family of Hamiltonians: Let    $  H_0$ satisfy   Assumption~\ref{assum:1} and   $[  H_0, S] =0$ and  assume $V\in\PO( \Df,\Hf)$.  Then,  for $\lambda_0>0$ sufficiently small, it holds that  
\[
H_\lambda :=      H_0 +  \lambda V\,
\]
is self-adjoint on the domain of $H_0$ and has a spectral gap 
  at $\mu$ for each $\lambda\in [0,\lambda_0)$.
As before we consider the gapped Fermi  projection $ \Pi_\lambda := \chi_{(-\infty,\mu]}( H_\lambda)$ of $  H_\lambda$, put $  \Pi_0^{(\ell)} :=  \Pi_0(\Id\otimes p_\ell)$ and the associated Chern numbers  Chern$(  \Pi_0^{(\ell)} )_{ij}$. The $\lambda$-independent $S$-Chern number associated with $\Pi_\lambda$ is again 
\[
\mbox{$S$-Chern}(  \Pi_\lambda  )_{ij}:=  \sum_{\ell=1}^k s_\ell \cdot  \mbox{Chern}(  \Pi_0^{(\ell)} )_{ij} \equiv \mbox{$S$-Chern}(  \Pi_0  )_{ij}\,.
\]

\begin{proposition}\label{prop:robust spin chern}
Let $H_\lambda= H_0+\lambda V$, be a perturbation of a spin-commuting Hamiltonian $H_0$ as defined  above. Then 
  the $S$-conductivity $\sigma_{ij,\lambda}^S $ of $H_\lambda$ satisfies
\[
\sigma_{ij,\lambda}^S = \mbox{\rm$S$-Chern}(\Pi_0)_{ij} + \Or(\lambda)\,.
\]
\end{proposition}

\begin{proof} By standard perturbation theory we obtain 
\[
\|\Pi_\lambda  -   \Pi_0 \|_{\PO(  \Hf,\Df )} 
 = \Or(\lambda)\,,\quad\|[\Pi_\lambda  -   \Pi_0  ,X_i]\|_{\PO(  \Hf,\Df )}= \Or(\lambda)  \,, \quad  \|[\Pi_\lambda ,S]\|_{\PO(  \Hf,\Df )} 
 = \Or(\lambda)\,.
\]
Observe that $\norm{\Pi_\lambda}_{1,\tau}=\norm{\Pi_0}_{1,\tau}$ using the smallness argument in $^{\ref{fn:B_1tau prop for pert Pi}}$. 
Hence,
starting from  \eqref{eq: per sigma} and \eqref{eq: nonper sigma}, we find   
\begin{align*}
\sigma_{ij,\lambda}^S &=  \frac{1}{\abs{\FC_1}}\re \Tr \Big( \chi_1\(\iu \Pi_\lambda\big[[X_i,\Pi_\lambda], S\su{D}_\lambda [X_j,\Pi_\lambda] \big] +\iu [H_\lambda,X_{i,\lambda}\su{D}] S\su{OD}_\lambda \Pi_{1,\lambda}  \) \chi_1\Big)  \\
& \quad + \frac{1}{\abs{\FC_1}}\re \Tr  \Big(\chi_1 \iu X_{i,\lambda}\su{D}  [H_\lambda, S ]  \Pi_1   \chi_1 \Big).
 \\
&=\frac{\iu}{\abs{\FC_1}}\re \Tr \Big( \chi_1    \Pi_0\big[[X_i, \Pi_0],   S  [X_j, \Pi_0] \big]    \chi_1\Big) + \Or(\lambda) \\
&= \frac{\iu}{\abs{\FC_1}} \sum_\ell s_\ell \re \Tr \Big( \chi_1    \Pi_0\big[[X_i, \Pi_0],  ( \Id\otimes p_{\ell}) [X_j, \Pi_0] \big]    \chi_1\Big) + \Or(\lambda) \\
&= \frac{\iu}{\abs{\FC_1}} \sum_\ell s_\ell  \Tr \Big( \chi_1   \Pi_0^{(\ell)} \big[[X_i, \Pi_0^{(\ell)} ],     [X_j, \Pi_0^{(\ell)} ] \big]    \chi_1\Big) + \Or(\lambda) \\
&= \mbox{\rm$S$-Chern}(\Pi_0)_{ij} + \Or(\lambda)\,. 
\end{align*}
In the second to last equality we used that
$\Id\otimes p_{\ell}$  commutes with~$X_i$ and with $  \Pi_0$.

\end{proof}


\section{Proofs} \label{sec:Proof}

In this Section we present the proofs for a number of core technical results used in the body of the paper. We will exploit the operator spaces defined in Section~\ref{sec:H0} and in particular their properties with respect to the trace per unit volume.

\subsection{Diagonal and off-diagonal operators} \label{sec:OD}

The Fermi projection $\Pi_0$ of the unperturbed Hamiltonian $H_0$ clearly induces a decomposition of $L^2(\mathcal{X}) \otimes \C^N$ into $\Ran \Pi_0 \oplus (\Ran \Pi_0)^\perp$. Correspondingly, operators acting in $L^2(\mathcal{X}) \otimes \C^N$ will admit a block decomposition. We review in this Section some properties of this decomposition, heading towards the proof of a well-known formula from asymptotic perturbation theory which allows to invert the Liouvillian  
$[H_0, \, \cdot\,]$ acting on operators which only have off-diagonal blocks.

\begin{definition}[Diagonal and off-diagonal parts]
Given an operator $A$ and an orthogonal projection $\Pi$, \ie $\Pi = \Pi^* = \Pi^2$, such that $A\,\Pi$ is densely defined
\footnote{{The operator $A$ may be unbounded and thus a careful analysis is required. In particular, we want to avoid pathological examples and have $A \supseteq A\su{D} + A\su{OD}$ on a \emph{dense} domain. Later, we will for example apply the block decomposition with respect to the Fermi projection $\Pi_0$ to the operator $X_i$ for $i\in\set{1, \ldots, d}$ (see Subsection~\ref{ssec:linresponsecoeff}), and $X_i \, \Pi_0$ is densely defined under Assumption~\ref{assum:1} on the unperturbed Hamiltonian (compare Corollary~\ref{cor:der Pi0 and H0}).  }
}, 
one defines its \emph{diagonal} and \emph{off-diagonal parts} as
\begin{align*}
A\su{D} & := \Pi A \Pi + (\Id - \Pi) A (\Id - \Pi), \\
A\su{OD} & := \Pi A (\Id - \Pi) + (\Id - \Pi) A \Pi,
\end{align*}
respectively. The operator $A$ is called \emph{diagonal} (resp. \emph{off-diagonal}) if $A = A\su{D}$ (resp. $A = A\su{OD}$).
\end{definition}

We collect in the following Lemma two simple properties of diagonal and off-diagonal operators in a general Hilbert space $\Hi$.

\begin{lemma} \label{lemma:diag}
Let $A$ be an operator acting in $\Hi$ and $\Pi$ an orthogonal projection on $\Hi$ such that $A \Pi$ is densely defined. 
\begin{enumerate}[label=(\roman*), ref=(\roman*)]
 \item \label{item:diag1} $A$ is diagonal if and only if $[A,\Pi] = 0$. $A$ is off-diagonal if and only if 
 $A = A \Pi + \Pi A$.
 \item \label{item:diag2} $A\su{OD} = [[A,\Pi],\Pi]$. 
\end{enumerate}
\end{lemma}
\begin{proof}
{\it \ref{item:diag1}} If $A$ is diagonal then $\Pi A = \Pi A\su{D} = \Pi A \Pi = A \Pi$, as $\Pi^2 = \Pi$ and $\Pi (\Id - \Pi) = 0$. Conversely, if $[A, \Pi] = 0$ then $\Pi A (\Id - \Pi) = A \Pi (\Id - \Pi) = 0$ and similarly $(\Id - \Pi) A \Pi = 0$, so that $A\su{OD}=0$.

If $A$ is off-diagonal, \ie $A = \Pi A (\Id - \Pi) + (\Id - \Pi) A \Pi$, then applying $\Pi$ to the left and to the right of this equality we get
\[ \Pi A = \Pi A (\Id - \Pi), \quad A \Pi = (\Id - \Pi) A \Pi, \]
and hence the conclusion follows. Conversely, if $A = A \Pi + \Pi A$, then 
\[ \Pi A \Pi = 2 \, \Pi A \Pi \quad \Longrightarrow \quad \Pi A \Pi = 0. \]
Moreover, $(\Id - \Pi) A (\Id - \Pi) = 0$ as $\Pi (\Id - \Pi) = (\Id - \Pi) \Pi = 0$. Hence $A\su{D}=0$.

{\it \ref{item:diag2}} We have
\begin{align*}
[[A,\Pi],\Pi] & = [A \Pi - \Pi A, \Pi] = A \Pi - \Pi A \Pi - \Pi A \Pi + \Pi A \\
& = (\Id - \Pi) A \Pi + \Pi A (\Id - \Pi). \qedhere
\end{align*}
\end{proof}

\subsection{Inverse Liouvillian} \label{sec:ProofLiouvillian}

We study here the Liouvillian (super-)operator $B \mapsto [H_0,B]$ associated to the unperturbed Hamiltonian, and in particular the possibility to invert it away from its kernel. We look in other words for the solution $B$ to the equation $[H_0,B] = A$, where $A\in \PO(\Hf)$  is off-diagonal with respect to the decomposition $\Hi = \Pi_0 \Hi \oplus (\Id-\Pi_0)\Hi$. We state in the following Proposition the solution to this problem, which traces back at least to \cite{Nenciu93,Nenciu02}. 

\begin{proposition} \label{prop:I(A)}
Under Assumption~\ref{assum:1}, let $A\in \PO(\Hf)$ be such that $A=A\su{OD}$ with respect to $\Pi_0$. Then the operator
$\mathcal{I}(A)$, defined in \eqref{eqn:I(A)},
is the unique off-diagonal solution in $\PO(\Hf,\Df)$ to the equation
\begin{equation} \label{eqn:Liouvillian}
[H_0, \mathcal{I}(A)] = A \qquad  \text{on} \quad \UZ^{-1}L^2_{\varrho}(\R^d,\Df).
\end{equation}
\end{proposition}
\begin{proof}
From the very definition \eqref{eqn:I(A)} and our hypotheses on $A$, we have that $\mathcal{I}(A)$ is off-diagonal and is in $\PO(\Hf,\Df)$ by Lemma~\ref{lem:resolvent-projectionPHfDf}.

Thus we need only to prove \eqref{eqn:Liouvillian}. Since $H_0(k)\in \mathcal{L}(\Df,\Hf)$ and $\Ran(H_0(k) - z \Id)^{-1}\subset \Df$ for any $z\in \rho(H_0)$, applying \cite[\S V.5 Corollary 2]{Yosida68} we have that on $ \Df$
\[
[H_0(k), \mathcal{I}(A)(k)]=\frac{\iu}{2 \pi} \oint_C \di z \,[H_0(k)  , (H_0(k) - z \Id)^{-1} \, [A(k),\Pi_0(k)] \, (H_0(k) - z \Id)^{-1}].
\]	
Hence, we obtain that on the domain $\UZ^{-1}L^2_{\varrho}(\R^d,\Df)$
\begin{align*}
[H_0, \mathcal{I}(A)]&=\frac{\iu}{2 \pi} \oint_C \di z \,[H_0 - z \Id , (H_0 - z \Id)^{-1} \, [A,\Pi_0] \, (H_0 - z \Id)^{-1}]\\
&=\frac{\iu}{2 \pi} \oint_C \di z \,[[A,\Pi_0], (H_0 - z \Id)^{-1}]=[[A,\Pi_0],\Pi_0]=A\su{OD}=A,
\end{align*}
using the Riesz formula (compare~\eqref{eqn:Riesz}) and Lemma~\ref{lemma:diag}\ref{item:diag2}.

Finally, notice that $\mathcal{I}(A)$ is the unique off-diagonal solution in $\PO(\Hf,\Df)$ to equation~\eqref{eqn:Liouvillian} for any off-diagonal operator $A\in \PO(\Hf)$. Indeed, if $B\in \PO(\Hf,\Df)$ is another solution to \eqref{eqn:Liouvillian}, then $\mathcal{I}(A)-B$ commutes with $H_0$, and hence with $\Pi_0$. By Lemma~\ref{lemma:diag}\ref{item:diag1}, $\mathcal{I}(A)-B$ is diagonal, and hence $B=B\su{OD} = \mathcal{I}(A)\su{OD} = \mathcal{I}(A)$.
\end{proof}

\subsection{NEASS} \label{sec:ProofNEASS}

This Section is devoted to the proof of Proposition~\ref{prop:Pi1} and thus to the explicit construction of the NEASS $\Pi^\epsi$ satisfying \ref{item:SA1} and \ref{item:SA3}. In order to give this proof, we first need the following preparatory Lemma.

\begin{lemma}
\label{lem:(E^{eps A} -Id) and [E^{eps A},Xj] PO(Hf,Df) unif in eps}
If $A\in\PO(\Hf,\Df)$, then both $(\E^{\eps A} -\Id)$ and $\overline{[\E^{\eps A},X_j]}$ are in $\PO(\Hf,\Df)$ and their norms in this space are bounded uniformly in $\eps\in[0,1]$.
\end{lemma}
\begin{proof}
Clearly, $(\E^{\eps A} -\Id)$ is periodic as $A$ is periodic. Since $\mathcal{L}(\Hf,\Df)$ is a Banach space, 
\[
\E^{\eps A}(k)-\Id=\E^{\eps A(k)}-\Id=\sum_{n=1}^\infty \frac{\eps^n A^n(k)}{n!}
\]
 converges in $\mathcal{L}(\Hf,\Df)$ uniformly in $k\in K$ for any compact set $K\subset\R^d$ and $\epsi\in[0,1]$, as the sequence $\{\sum_{n=1}^N \frac{\eps^n A^n(k)}{n!}\}_{N\in \N}\subset\mathcal{L}(\Hf,\Df)$ converges absolutely in $\mathcal{L}(\Hf,\Df)$ uniformly in $k\in K$ and $\eps\in [0,1]$. Moreover, observe that each summand is such that
\[
\R^d\ni k\mapsto\frac{A^n(k)}{n!}\in\mathcal{L}(\Hf,\Df)\text{ is smooth }\forall n\geq 1
\] 
and 
\[
\sum_{n=1}^\infty \frac{\epsi^n\partial_{k_j}\left(A^n(k)\right)}{n!}=\sum_{n=1}^\infty \frac{\eps^n}{n!}\sum_{h=0}^{n-1} A^h(k) \, \big( \partial_{k_j} A(k) \big)\, A^{n-1-h}(k)
\] 
converges in $\mathcal{L}(\Hf,\Df)$ uniformly in $\epsi\in[0,1]$ and in $k\in K$ for any compact set $K\subset\R^d$ due to the assumption that $A\in\PO(\Hf,\Df)$.
Therefore we are allowed to interchange the derivation in $k$ and the series in $n$. Iterating this argument implies that
$(\E^{\eps A}-\Id)$ is in $\PO(\Hf,\Df)$ and that its norm in this space is uniformly bounded with respect to $\epsi\in[0,1]$. Thus by Proposition~\ref{prop:derivation of P spaces} we deduce that $\overline{[\E^{\eps A},X_j]}=\overline{[\E^{\eps A}-\Id , X_j]}\in \PO(\Hf,\Df)$ again with uniform bounds on its norm for $\eps\in[0,1]$.
\end{proof}

We are now ready to tackle the

\begin{proof}[Proof of Proposition~\ref{prop:Pi1}]
In this proof, we will abbreviate the expression ``the map $[0,1] \ni \eps \mapsto A^\eps \in \PO$ is uniformly bounded'' for some space of operators $\PO$ by just saying that ``$A^\eps$ is in $\PO$ uniformly in $\eps \in [0,1]$''.

{\it \ref{item:propSA1}} By Corollary~\ref{cor:der Pi0 and H0}\ref{item:[Pi0,X_j]} one has that $[\overline{[X_j,\Pi_0]},\Pi_0]$ is in $\PO(\Hf,\Df)\subset \PO(\Hf)$ and is off-diagonal with respect to $\Pi_0$, hence Proposition~\ref{prop:I(A)} implies that $\mathcal{S}\in\PO(\Hf,\Df)$. Self-adjointness of $\mathcal{S}$  is evident.

{\it \ref{item:propSA3}} By Taylor's formula we find that for any $\eps>0$
\[
\E^{-\I \epsi \mathcal{S}} \,\Pi_0\, \E^{\I\epsi \mathcal{S}} = \Pi_0 + \I\eps [\Pi_0,\mathcal{S}]  - \tfrac{\epsi^2}{2}  \E^{-\I \tilde \epsi \mathcal{S}} \,[\mathcal{S},[\mathcal{S},\Pi_0]]\, \E^{\I\tilde \epsi \mathcal{S}}
\]
for some $\tilde \eps\in (0,\eps)$.
Thus, in view of Lemma \ref{lemma:diag}.\ref{item:diag2} and of the fact that $\overline{[  X_j,\Pi_0]}={\overline{[  X_j,\Pi_0]}}\su{OD}$, one has
\[
\Pi_1 = \I [\Pi_0,\mathcal{S}]  =  -\left[\Pi_0,\mathcal{I} \( \big[\, \overline{[X_j,\Pi_0]}\,,\Pi_0 \big]\) \right] =   \mathcal{I}\( \big[\big[\, \overline{[X_j,\Pi_0]}\,,\Pi_0\big],\Pi_0\big] \)= \mathcal{I}\( \overline{[  X_j,\Pi_0]} \).
\]
Moreover, 
\begin{equation}
\label{eqn:Pireps}
\Pi_r^\eps=\tfrac{\I}{2}  \E^{-\I \tilde \epsi \mathcal{S}} \,[\Pi_1,\mathcal{S}]\, \E^{\I\tilde \epsi \mathcal{S}} .
\end{equation}
In view of Corollary~\ref{cor:der Pi0 and H0}\ref{item:[Pi0,X_j]} and of Proposition~\ref{prop:I(A)}, we have $\Pi_1\in\PO(\Hf,\Df)$.
Notice now that $[\Pi_1,\mathcal{S}]$ is in $\PO(\Hf,\Df)$ (because $\Pi_1, \mathcal{S}\in \PO(\Hf,\Df)$), and $(\E^{-\I \tilde \epsi \mathcal{S}}-\Id)\in\PO(\Hf,\Df)$ uniformly in $\tilde \eps\in (0,\eps)\subseteq [0,1]$ by Lemma~\ref{lem:(E^{eps A} -Id) and [E^{eps A},Xj] PO(Hf,Df) unif in eps}. Therefore we conclude that
\begin{equation}
\label{eqn:Pi_r^eps in PO(Hf,Df)}
\Pi_r^\eps=\tfrac{\I}{2}  (\E^{-\I \tilde \epsi \mathcal{S}}-\Id) \,[\Pi_1,\mathcal{S}]\, \E^{\I\tilde \epsi \mathcal{S}}+\tfrac{\I}{2}  [\Pi_1,\mathcal{S}]\, \E^{\I\tilde \epsi \mathcal{S}}\in\PO(\Hf,\Df)
\end{equation}
uniformly in $\tilde{\epsi} \in (0,\eps)\subseteq [0,1]$.
Finally, on the domain $\UZ^{-1}{C^\infty_\varrho(\R^d,\Df)}$ we have that
\begin{eqnarray*}
[H^\eps,\Pi^\eps] = \eps \left( [H_0,\Pi_1]-[X_j,\Pi_0]\right) + \epsi^2 \left([H^\epsi, \Pi_r^\eps]-[X_j,\Pi_1]\right).
\end{eqnarray*}
On the right-hand side, the first term vanishes in view of equation~\eqref{eqn:Liouvillian}:
\[
 [H_0,\Pi_1] = \left[ H_0,\mathcal{I} \( \overline{[  X_j,\Pi_0]} \) \right] =   [  X_j,\Pi_0]\text{  on $\UZ^{-1}{C^\infty_\varrho(\R^d,\Df)}$}.
\]
As for the second term, we recognize that 
\begin{equation}
\label{eqn:Reps}
 [H^\epsi, \Pi_r^\eps]-[X_j,\Pi_1]\Big|_ {\UZ^{-1}{C^\infty_\varrho(\R^d,\Df)}}
\end{equation}
extends to a bounded operator in $\PO(\Hf)$ uniformly in $\eps\in[0,1]$.
Indeed, the second summand $[X_j,\Pi_1]\Big|_ {\UZ^{-1}{C^\infty_\varrho(\R^d,\Df)}}$ in \eqref{eqn:Reps} extends to an operator in $\PO(\Hf,\Df)\subset \PO(\Hf)$ by Proposition~\ref{prop:derivation of P spaces}. We split instead the first summand in \eqref{eqn:Reps} as
\[
[H_0,\Pi_r^\eps]-\eps[X_j,\Pi_r^\eps]\Big|_ {\UZ^{-1}{C^\infty_\varrho(\R^d,\Df)}}.
\] 
The first of the terms above satisfies, in view of \eqref{eqn:Pi_r^eps in PO(Hf,Df)},
\[
\norm{\overline{[H_0,\Pi_r^\eps]}}_{\PO(\Hf)}\leq 2 \norm{H_0}_{\PO(\Df,\Hf)}\norm{\Pi_r^\eps}_{\PO(\Hf,\Df)}\text{ for }\eps\in[0,1],
\]
while $[X_j,\Pi_r^\eps]\Big|_ {\UZ^{-1}{C^\infty_\varrho(\R^d,\Df)}}$  extends to an operator in $\PO(\Hf)$ applying again Proposition~\ref{prop:derivation of P spaces}. It remains only to show that $\norm{\overline{[X_j,\Pi_r^\eps]}}_{\PO(\Hf)}$ is bounded uniformly in $\eps\in[0,1]$.
By Leibniz's rule 
\begin{align*}
2\norm{\overline{[X_j,\Pi_r^\eps]}}_{\PO(\Hf)}\leq& \norm{[\Pi_1,\mathcal{S}] }_{\PO(\Hf,\Df)}\(\norm{\overline{[X_j,\E^{-\I \tilde \epsi \mathcal{S}} ]}}_{\PO(\Hf,\Df)}+\norm{\overline{[X_j,\E^{\I \tilde \epsi \mathcal{S}} ]}}_{\PO(\Hf,\Df)}\)\\
&+\norm{\overline{[X_j,[\Pi_1,\mathcal{S}]]} }_{\PO(\Hf,\Df)},
\end{align*}
where the right-hand side is bounded uniformly in $\eps\in[0,1]$ by Lemma~\ref{lem:(E^{eps A} -Id) and [E^{eps A},Xj] PO(Hf,Df) unif in eps} and Proposition~\ref{prop:derivation of P spaces}.
This concludes the proof.
\end{proof}

\subsection{{Well-posedness of the proper $S$-conductivity}} 
\label{sec:Proofsigma}
Here we prove that $\re\tau(J\subm{prop}{i}^S\,\Pi_1)=\sigma\subm{prop}{ij}^S$, by using \eqref{eqn:sigmaprop} and  combining the following results.
\begin{lemma} 
\label{lem:Pi in tau-class}
Under Assumption \ref{assum:1} we have that $\Pi_0,\Pi_1,\Pi_r^\eps\in \B_1^\tau.$
\end{lemma}
\begin{proof}
By {Assumption~\ref{assum:1}} $\Pi_0(k)$ is a finite-rank projection on $\Hf$ with rank $m$ (independent of $k$). Thus, in view of Proposition~\ref{prop:charge-tauPeriodic}, we get 
\begin{equation} \label{190708}
\tau(\abs{\Pi_0})=\tau(\Pi_0)=\frac{1}{\abs{\FC_1}}\Tr(\chi_1\Pi_0 \chi_1)=\frac{1}{\abs{\FC_1}}\frac{1}{|\mathbb{B}^d|} \int_{\mathbb{B}^d}\di k\,\Tr_{\Hf}(\Pi_0(k))=\frac{m}{\abs{\FC_1}},
\end{equation}
hence $ \Pi_0 \in \B_1^\tau$.  In view of Propositions~\ref{prop:Pi1} and~\ref{prop:I(A)}, we have that $\Pi_1=\Pi_1\su{OD}$ is off-diagonal with respect to the orthogonal decomposition induced by $\Pi_0$, and hence Lemma~\ref{lemma:diag}\ref{item:diag1} implies that
\[
\Pi_1=\Pi_1\Pi_0+\Pi_0\Pi_1\,.
\]
Since $\Pi_1\in \PO(\Hf)\subset \B_\infty^\tau$ and $\Pi_0\in \B_1^\tau$ by \eqref{190708}, we deduce that the right-hand side of the above is in $\B_1^\tau$ as $\B_\infty^\tau\cdot\B_1^\tau\cdot \B_\infty^\tau \subset \B_1^\tau $.

Finally, recall from \eqref{eqn:Pireps} that
$
\Pi_r^\eps =\tfrac{\I}{2}  \E^{-\I \tilde \epsi \mathcal{S}} \,[\Pi_1,\mathcal{S}]\, \E^{\I\tilde \epsi \mathcal{S}}  
$
for some $\tilde \eps\in (0,\eps)$. As we just proved $\Pi_1\in \B_1^\tau$, and the other operators which appear on the right-hand side of the above are in $\PO(\Hf)\subset \B_\infty^\tau$, we conclude that $\Pi_r^\eps\in\B_1^\tau$.
\end{proof}

\begin{proposition} 
\label{cor:welldef sigmaijeps}
Under Assumption~\ref{assum:1} and the hypotheses on $S$ in Definition~\ref{def:$S$-current}, for $\Pi_\sharp\in\{ \Pi_0,\Pi_1,\Pi_r^\eps\}$ we have that
\begin{enumerate}[label=(\roman*), ref=(\roman*)]
\item \label{item:welldef sigmaijeps 1} the operators $ \overline{[H_0,X_i]} S \Pi_\sharp$ and  $[H_0, S]\Pi_\sharp$ are in $ \B_1^\tau$,
\item \label{item:welldef sigmaijeps 2}  the operators $ X_i [H_0, S]\Pi_\sharp$ and $\overline{[H_0,X_i]}S\Pi_\sharp$ 
have finite trace per unit volume.
\end{enumerate}
\end{proposition}
\begin{proof}
{\it \ref{item:welldef sigmaijeps 1}}
We have that $\Pi_\sharp\in \PO(\Hf,\Df)$ by Proposition~\ref{prop:Pi1},  $S\in \PO( \Df)$ by hypothesis, and $\overline{[H_0,X_i]}\in \PO(\Df,\Hf)$ by Corollary~\ref{cor:der Pi0 and H0}\ref{item:[H0,X_j]}, thus we deduce that
\begin{equation}
\label{eqn:[H0,Xi]SPi}
\overline{[H_0, X_i ]} S\Pi_\sharp\in \PO(\Df,\Hf)\cdot \PO(\Df)\cdot \PO(\Hf,\Df)\subset \PO(\Hf)\subset \B_\infty^\tau.
\end{equation}
On the other hand, since
\begin{equation}
\label{eqn:[H0,S] in PO(Df,Hf)}
[H_0,S]=H_0 S-S H_0\in\PO(\Df,\Hf)\cdot \PO(\Df)+\PO(\Hf)\cdot \PO(\Df,\Hf)\subset \PO(\Df,\Hf),
\end{equation}
we get that
\begin{equation}
\label{eqn:[H_0,S]Pi}
[H_0, S]\Pi_\sharp\in \PO(\Df,\Hf)\cdot \PO(\Hf,\Df)\subset \PO(\Hf)\subset \B_\infty^\tau.
\end{equation}
Now we are going to show that the operators in \eqref{eqn:[H0,Xi]SPi} and \eqref{eqn:[H_0,S]Pi} are in $\B_1^\tau$, using the previous results.

First we analyse the case $\Pi_\sharp=\Pi_0$. In view of Lemma~\ref{lem:Pi in tau-class} we have $\Pi_0\in \B_1^\tau$, thus we deduce that
\[
\overline{[H_0, X_i ]} S\Pi_0=\overline{[H_0, X_i ]} S\Pi_0\cdot\Pi_0\in \B_\infty^\tau\cdot \B_1^\tau\subset  \B_1^\tau
\]
and similarly
\[
[H_0, S]\Pi_0=[H_0, S]\Pi_0\cdot\Pi_0 \in \B_\infty^\tau\cdot\B_1^\tau\subset  \B_1^\tau.
\]

We proceed with the case $\Pi_\sharp=\Pi_1$. By virtue of Lemma~\ref{lem:Pi in tau-class}, of the construction in Proposition~\ref{prop:Pi1}, and of Lemma~\ref{lemma:diag}\ref{item:diag1}, we have that $\B_1^\tau\ni \Pi_1=\Pi_1\Pi_0+\Pi_0\Pi_1$, hence we obtain that
\begin{equation}
\label{eqn:[H_0,Xi ]SPi1 in B1tau}
\overline{[H_0, X_i ]} S\Pi_1=\overline{[H_0, X_i ]} S\Pi_1\cdot\Pi_0+\overline{[H_0, X_i ]} S\Pi_0\cdot\Pi_1\in \B_\infty^\tau\cdot \B_1^\tau\subset  \B_1^\tau.
\end{equation}
One can argue in an analogous way to conclude that $[H_0, S]\Pi_1\in \B_1^\tau$ using \eqref{eqn:[H0,S] in PO(Df,Hf)}.

Finally, we analyse the case $\Pi_\sharp=\Pi_r^\eps=\tfrac{\I}{2}  \E^{-\I \tilde \epsi \mathcal{S}} \,[\Pi_1,\mathcal{S}]\, \E^{\I\tilde \epsi \mathcal{S}}$. Notice that
\begin{equation}\label{eqn:[H0,Xi]SPireps}
\begin{aligned}
\overline{[H_0, X_i ]} S\Pi_r^\eps&=\tfrac{\I}{2}\overline{[H_0, X_i ]} S(\E^{-\I \tilde \epsi \mathcal{S}}-\Id)  [\Pi_1,\mathcal{S}]\, \E^{\I\tilde \epsi \mathcal{S}}+\tfrac{\I}{2}\overline{[H_0, X_i ]} S [\Pi_1,\mathcal{S}]\, \E^{\I\tilde \epsi \mathcal{S}}\\
&=\overline{[H_0, X_i ]} S(\E^{-\I \tilde \epsi \mathcal{S}}-\Id)  \E^{\I\tilde \epsi \mathcal{S}} \Pi_r^\eps+\tfrac{\I}{2}\overline{[H_0, X_i ]} S \Pi_1 \mathcal{S} \E^{\I\tilde \epsi \mathcal{S}}-\tfrac{\I}{2}\overline{[H_0, X_i ]} S \mathcal{S} \Pi_1 \E^{\I\tilde \epsi \mathcal{S}}.
\end{aligned}
\end{equation}
Observe that on the right-hand side of the last equality each summand is in $\B_1^\tau$. Indeed, since $(\E^{-\I \tilde \epsi \mathcal{S}}-\Id) \in \PO(\Hf,\Df)$ by Lemma~\ref{lem:(E^{eps A} -Id) and [E^{eps A},Xj] PO(Hf,Df) unif in eps} and $\E^{\I\tilde \epsi \mathcal{S}} \Pi_r^\eps\in\B_1^\tau$ by Lemma~\ref{lem:Pi in tau-class}, we deduce that
\[
\overline{[H_0, X_i ]} S(\E^{-\I \tilde \epsi \mathcal{S}}-\Id) \cdot \E^{\I\tilde \epsi \mathcal{S}} \Pi_r^\eps\in \PO(\Df,\Hf)\cdot \PO(\Df)\cdot \PO(\Hf,\Df)\cdot\B_1^\tau \subset \PO(\Hf) \cdot \B_1^\tau\subset \B_1^\tau.
\]
Using \eqref{eqn:[H_0,Xi ]SPi1 in B1tau} and $\mathcal{S} \E^{\I\tilde \epsi \mathcal{S}}\in \PO(\Hf)$ by Proposition~\ref{prop:Pi1}\ref{item:propSA1}, we infer that
\[
\overline{[H_0, X_i ]} S \Pi_1\cdot \mathcal{S} \E^{\I\tilde \epsi \mathcal{S}}\in  \B_1^\tau\cdot \B_\infty^\tau\subset \B_1^\tau.
\]
As for the third summand in \eqref{eqn:[H0,Xi]SPireps}, observe that $\mathcal{S}\in\PO(\Hf,\Df)$ by Proposition~\ref{prop:Pi1}\ref{item:propSA1} and $\Pi_1 \E^{\I\tilde \epsi \mathcal{S}}\in \B_1^\tau$ by Lemma~\ref{lem:Pi in tau-class}, therefore we get that
\[
\overline{[H_0, X_i ]} S \mathcal{S} \cdot \Pi_1 \E^{\I\tilde \epsi \mathcal{S}}\in \PO(\Df,\Hf)\cdot \PO(\Df)\cdot \PO(\Hf,\Df)\cdot\B_1^\tau \subset \PO(\Hf) \cdot \B_1^\tau\subset \B_1^\tau.
\]
From a similar argument it follows that $[H_0, S]\Pi_r^\eps\in\B_1^\tau$ owing to \eqref{eqn:[H0,S] in PO(Df,Hf)}. 

{\it \ref{item:welldef sigmaijeps 2}} 
The conclusion follows by applying part~\ref{item:welldef sigmaijeps 1} of the statement (which we just proved), 
Lemma~\ref{lem:tau cont funct1andL}, Proposition~\ref{prop:charge-tauPeriodic}\ref{item:per+traceclassoncompact}, and Proposition~\ref{prop:AXi}\ref{it:AXi exh dep}.
\end{proof}

In the following Lemma, we prove that some expectation values of $J\subm{conv}{i}^S$ which 
are relevant to transport theory, reduce to the real part of the 
expectation value of $J\subm{naive}{i}^S := \iu\overline{[H_0, X_i ]} S$.  

\begin{lemma} \label{Lem:reduce_Jconv}
Under Assumption~\ref{assum:1} and the hypotheses on $S$ in Definition~\ref{def:$S$-current}, 
let $\Pi_\sharp\in\{ \Pi_0,\Pi_1,\Pi_r^\eps, \Pi^\eps \}$.  Then
\begin{equation} \label{eq:reduce_Jconv}
\tau(J\subm{conv}{i}^S \,\,\Pi_\sharp)=\re\tau(\iu\overline{[H_0, X_i ]} S \,\,\Pi_\sharp).
\end{equation}
\end{lemma}

\begin{proof} 
We first prove the claim for $\Pi_\sharp\in\{ \Pi_0, \Pi^\eps \}$,  exploiting the fact that 
$\Pi_\sharp^2 = \Pi_\sharp = \Pi_\sharp^*$ in this case. Using in addition that  $\Pi_\sharp \in \B_1^\tau$
by Lemma \ref{lem:Pi in tau-class}, $J\subm{naive}{i}^S \Pi_\sharp \in \B_{\infty}^\tau$ by Proposition~\ref{cor:welldef sigmaijeps}\ref{item:welldef sigmaijeps 1}, the cyclicity of the trace per unit volume, $\overline{\Tr(A)}=\Tr(A^*)$, one has that 
\begin{align*}
\re\tau(\iu\overline{[H_0,X_i]}S\,\Pi_\sharp)&=  \re \tau(\Pi_\sharp \, \iu\overline{[H_0,X_i]}S \,\Pi_\sharp) 
\cr
&=\frac{1}{2}\Big(  \tau(\Pi_\sharp \, J\subm{naive}{i}^S \,\Pi_\sharp)+  \tau(\Pi_\sharp  (J\subm{naive}{i}^S)^* \,\Pi_\sharp) \Big) \cr
&=   \tau(\Pi_\sharp \, J\subm{conv}{i}^S \,\Pi_\sharp) =   \tau(J\subm{conv}{i}^S \,\Pi_\sharp).
\end{align*}
The case $\Pi_\sharp = \Pi_1$ is subtler, and we crucially use the fact that $\Pi_1=\Pi_1\Pi_0+\Pi_0\Pi_1$,
and $\Pi_1^* = \Pi_1$.  Indeed, by using Lemma~\ref{prop:cycl of tau}, Lemma~\ref{lem:Pi in tau-class}, 
and Proposition~\ref{cor:welldef sigmaijeps} we obtain that
\begin{align*}
\re\tau(J\subm{naive}{i}^S\,\Pi_1)&=\re\tau(J\subm{naive}{i}^S\,\Pi_1\Pi_0)+\re\tau(J\subm{naive}{i}^S\,\Pi_0\Pi_1)\cr
&=\re\tau(\Pi_0J\subm{naive}{i}^S\,\Pi_1)+\re\tau(\Pi_1J\subm{naive}{i}^S\,\Pi_0)\cr
&=\frac{1}{2}\Big(  \tau(\Pi_0 J\subm{naive}{i}^S\,\Pi_1)+  \tau(\Pi_1 J\subm{naive}{i}^S\,\Pi_0)     \Big)\cr
&\phantom{=}+\frac{1}{2}\Big(  \overline{\tau(\Pi_0 J\subm{naive}{i}^S\,\Pi_1)}+  \overline{\tau(\Pi_1 J\subm{naive}{i}^S\,\Pi_0)}     \Big)\cr
&=  \frac{1}{2}  \tau(\iu\overline{[H_0,X_i]}S\,\Pi_1)\cr                   
&\phantom{=} + \frac{1}{2}\Big(  {\tau(\Pi_1S\iu\overline{[H_0,X_i]}\Pi_0)}+  {\tau(\Pi_0S\iu\overline{[H_0,X_i]}\Pi_1)}     \Big)\cr
&= \frac{1}{2}  \tau(\iu\overline{[H_0,X_i]}S\,\Pi_1) + \frac{1}{2}  \tau(S\iu\overline{[H_0,X_i]}\Pi_1) \cr 
&=\tau(J\subm{conv}{i}^S\,\Pi_1).
\end{align*}
Finally, it remains to prove the claim for $\Pi_\sharp = \Pi_r^\epsi$. The latter follows by $\R$-linearity from the previous
cases, as $\Pi_r^\epsi = \Pi^\epsi - \Pi_0 - \epsi \Pi_1$.  
\end{proof}

\begin{proposition}[Bounds on the remainder terms]
\label{cor:welldef sigmaij}
Under Assumption~\ref{assum:1}, there exist $C_1,C_2 \in \R$ such that
\[
{\abs{\tau(J\subm{conv}{i}^S\,\Pi_r^\eps)}\leq C_1\quad\text{ and }\quad\abs{\tau(J\subm{prop}{i}^S\,\Pi_r^\eps)}\leq C_2\qquad\forall\,\eps\in[0,1].}
\]
\end{proposition}
\begin{proof}
We begin by showing the first inequality. In view of Lemma~\ref{Lem:reduce_Jconv}
and the triangle inequality, we obtain that 
\begin{equation}
\label{eqn:remainderJconv}
\begin{aligned}
\abs{\tau(J\subm{conv}{i}^S\,\Pi_r^\eps)}&\leq \abs{\tau(\iu\overline{[H_0, X_i ]} S\Pi_r^\eps)}\leq \abs{ \tau ( \overline{[H_0, X_i ]} S(\E^{-\I \tilde \epsi \mathcal{S}}-\Id)  \E^{\I\tilde \epsi \mathcal{S}} \Pi_r^\eps)}\cr
&\phantom{\leq}+\tfrac{1}{2}\abs{ \tau (\overline{[H_0, X_i ]} S \Pi_1 \mathcal{S} \E^{\I\tilde \epsi \mathcal{S}})}+\tfrac{1}{2}\abs{\tau (\overline{[H_0, X_i ]} S \mathcal{S} \Pi_1 \E^{\I\tilde \epsi \mathcal{S}})}.
\end{aligned}
\end{equation}
By using the inequalities in \eqref{eqn:Binftytau B1tau Binftytau is B1tau} and Lemma~\ref{lem:(E^{eps A} -Id) and [E^{eps A},Xj] PO(Hf,Df) unif in eps}, the first summand on the right-hand side can be bounded uniformly in $\eps$ as 
\[
\abs{ \tau ( \overline{[H_0, X_i ]} S(\E^{-\I \tilde \epsi \mathcal{S}}-\Id)  \E^{\I\tilde \epsi \mathcal{S}} \Pi_r^\eps)}\leq C\norm{\overline{[H_0, X_i ]}}_{\PO(\Df,\Hf)} \norm{S}_{\PO(\Df)} \norm{[\Pi_1,\mathcal{S}]}_{1,\tau},
\]
where $C$ is a constant independent of $\eps$.
Applying again the inequalities in~\eqref{eqn:Binftytau B1tau Binftytau is B1tau}, we get for the second and third summand in \eqref{eqn:remainderJconv}, respectively
\[
\abs{ \tau (\overline{[H_0, X_i ]} S \Pi_1 \mathcal{S} \E^{\I\tilde \epsi \mathcal{S}})}\leq \norm{\overline{[H_0, X_i ]} S \Pi_1}_{1,\tau} \norm{\mathcal{S}}  
\]
using Proposition~\ref{cor:welldef sigmaijeps}\ref{item:welldef sigmaijeps 1}
and
\[
\abs{\tau (\overline{[H_0, X_i ]} S \mathcal{S} \Pi_1 \E^{\I\tilde \epsi \mathcal{S}})}\leq\norm{\overline{[H_0, X_i ]}}_{\PO(\Df,\Hf)} \norm{S}_{\PO(\Df)} \norm{\mathcal{S}}_{\PO(\Hf,\Df)} \norm{\Pi_1 }_{1,\tau}
\]
in view of Lemma~\ref{lem:Pi in tau-class}. 

Now, to obtain the second inequality of the thesis, notice that 
\begin{equation}
\label{eqn:remainderJprop}
\abs{\tau(J\subm{prop}{i}^S\,\Pi_r^\eps)}\leq\abs{ \tau (\overline{[H_0, X_i ]} S\Pi_r^\eps)} +\abs{ \tau(X_i [H_0, S]\Pi_r^\eps) }. 
\end{equation}
The first summand on the right-hand side is bounded uniformly in $\eps$, as it is shown before, while for the second one we proceed as follows.
 Since $[H_0, S]\Pi_r^\eps$ is $\tau$-class by Proposition~\ref{cor:welldef sigmaijeps}\ref{item:welldef sigmaijeps 1}, {Proposition~\ref{prop:AXi}\ref{it:AXi exh dep}} implies that
\[
\abs{ \tau(X_i [H_0, S]\Pi_r^\eps) }=\frac{1}{\abs{\FC_1}}\abs{\Tr\left( \chi_1 X_i [H_0, S]\Pi_r^\eps \chi_1 \right)}.
 \]
Applying the inequality $\abs{\Tr(AB)}\leq \norm{A}\Tr(\abs{B})$ for a bounded operator $A$ and a trace class operator $B$, estimate \eqref{eqn:tau cont funct} and Proposition~\ref{prop:charge-tauPeriodic}\ref{item:per+traceclassoncompact}, we have that
\begin{align*}
\frac{1}{\abs{\FC_1}}\abs{\Tr\left( \chi_1 X_i [H_0, S]\Pi_r^\eps \chi_1 \right)}
&\leq \frac{1}{\abs{\FC_1}}\norm{\chi_1 X_i \chi_1}  \Tr(\left \lvert \chi_1 [H_0, S]\Pi_r^\eps \chi_1  \right\rvert)\\
&\leq \norm{\chi_1 X_i \chi_1}\norm{[H_0, S]\Pi_r^\eps}_{1,\tau}.
\end{align*}
Finally, since $[H_0, S]\in\PO(\Df,\Hf)$ as shown in \eqref{eqn:[H0,S] in PO(Df,Hf)}, one can reason as in \eqref{eqn:remainderJconv} to conclude that $\norm{[H_0, S]\Pi_r^\eps}_{1,\tau}\leq D$, for a constant $D$  independent of $\eps$, which yields the second inequality in the thesis. 
\end{proof}

\subsection{Chern-like and extra contributions to the $S$-conductivity} 
\label{sec:ProofKE}

In this subsection we prove Proposition~\ref{prop:Extrachanged}, Proposition~\ref{prop:Btorque}, 
Theorem~\ref{thm:notKubo} and Lemma~\ref{lemma:Extra2}.

\begin{proof}[Proof of Proposition~\ref{prop:Extrachanged}]
We are going to show that $\mathsf{C}$ and $\mathsf{E}_\ell$ are in $\B_1^\tau$ for any $\ell\in\{1,\dots,4\}$ and then equality~\eqref{eqn:tau reads only prim cell} follows  from Proposition~\ref{prop:charge-tauPeriodic}\ref{item:per+traceclassoncompact}.
We begin by looking at the Chern-like term $\mathsf{C}$. We can write on $\Ran(\chi_L)$
\begin{equation} 
\label{eqn:K}
\mathsf{C} = \Pi_0\big[\, \overline{[X_i,\Pi_0]}, S\overline{[X_j,\Pi_0]}\, \big],
\end{equation}
where in the last equality we have used Remark~\ref{rmk:SmoothBF} and \eqref{eqn:closure der of Pi0 and H0}. 
In view of Corollary~\ref{cor:der Pi0 and H0}\ref{item:[Pi0,X_j]} and Lemma~\ref{lem:Pi in tau-class}, we have that 
\begin{align*}
\Pi_0\cdot\big[\, \overline{[X_i,\Pi_0]}, S\overline{[X_j,\Pi_0]}\, \big] \in \B_1^\tau \,\cdot\,\PO(\Hf).
\end{align*}
Therefore, the above operators are $\tau$-class as $\PO(\Hf)\,\cdot\,\B_1^\tau\subset\B_1^\tau$. In view of the cyclicity of the trace per unit volume and the off-diagonality of $[X_i,\Pi_0]$, one can rewrite $\tau(\mathsf{C})$ as
\begin{align*}
\tau(\Pi_0\big[[X_i,\Pi_0], S [X_j,\Pi_0] \big])&=\tau(\Pi_0[X_i,\Pi_0] S [X_j,\Pi_0] )-\tau(\Pi_0S [X_j,\Pi_0][X_i,\Pi_0] )\\
&=\tau(\Pi_0[X_i,\Pi_0] S [X_j,\Pi_0] )-\tau(\Pi_0 [X_j,\Pi_0][X_i,\Pi_0] S)\\
&=\tau(\Pi_0\big[[X_i,\Pi_0] S, [X_j,\Pi_0]\big] ).
\end{align*}

We now analyse the first extra term $\mathsf{E}_1$. Similarly to the previous computation, we have on $\Ran(\chi_L)$
\begin{equation} 
\label{eqn:E1}
\mathsf{E}_1= \iu\Pi_0 \overline{[H_0,X_i]}\Pi_0 S \Pi_1 +\iu \Pi_0^{\perp}\overline{[H_0,X_i]}\Pi_0^{\perp} S \Pi_1.
\end{equation}
In view of Corollary~\ref{cor:der Pi0 and H0}\ref{item:[H0,X_j]}, Lemma~\ref{lem:resolvent-projectionPHfDf}, Lemma~\ref{lem:Pi in tau-class} and Proposition~\ref{cor:welldef sigmaijeps}\ref{item:welldef sigmaijeps 1}, we get that
\begin{align*}
&\Pi_0 \overline{[H_0,X_i]}\Pi_0S\cdot \Pi_1 \in  \PO(\Hf)\,\cdot\,\B_1^\tau\\
& \Pi_0^{\perp}\overline{[H_0,X_i]}\Pi_0^{\perp} S\Pi_1=\Pi_0^{\perp}\overline{[H_0,X_i]} S \Pi_1-\Pi_0^{\perp}\overline{[H_0,X_i]}\Pi_0 S \Pi_1\in  \B_1^\tau.
\end{align*}
Then, noticing that $\Pi_1=\Pi_1\su{OD}$ by construction (see Propositions~\ref{prop:Pi1}\ref{item:propSA3} and \ref{prop:I(A)}) and applying Lemma~\ref{prop:cycl of tau}, we obtain the final expression for $\tau(\mathsf{E}_1)$ in \eqref{eqn: tau E1 E3}.

We now move to $\mathsf{E}_2$. Analogously, we get on $\Ran(\chi_L)$
\begin{equation}
\label{eqn:E2}
\mathsf{E}_2=\iu H_0 \big[\,\overline{[X_i,\Pi_0]},\Pi_0 \,\big ]S\Pi_1-\iu  \big[\,\overline{[X_i,\Pi_0]},\Pi_0 \,\big ]S\Pi_1H_0.
\end{equation}
By Corollary~\ref{cor:der Pi0 and H0}\ref{item:[Pi0,X_j]}, Lemma~\ref{lem:Pi in tau-class} and Proposition~\ref{prop:Pi1}\ref{item:propSA3}, we deduce that
\begin{align*}
&H_0 \big[\,\overline{[X_i,\Pi_0]},\Pi_0 \,\big]S\cdot\Pi_1\in  \PO(\Hf)\,\cdot\,\B_1^\tau\\
&\big[\,\overline{[X_i,\Pi_0]},\Pi_0 \,\big ]\cdot S\Pi_1H_0\in\B_1^\tau\,\cdot\,\PO(\Hf).
\end{align*}
Now we analyse $\mathsf{E}_3$. Similarly, we obtain on $\Ran(\chi_L)$
\begin{equation} 
\label{eqn:E3}
\mathsf{E}_3=\iu \big[\,\overline{[X_i,\Pi_0]},\Pi_0 \,\big ]\cdot [S,H_0]\Pi_1\in \B_1^\tau\,\cdot\,\PO(\Hf).
\end{equation}
Analogously, for $\mathsf{E}_4$ we have that on $\Ran(\chi_L)$
\begin{equation} 
\label{eqn:E4}
\mathsf{E}_4=\iu\big[\overline{[X_i,\Pi_0]},\Pi_0 S \overline{[\Pi_0,X_j]} \big]\in \B_1^\tau.
\end{equation}
Finally, by applying  Proposition~\ref{cor:welldef sigmaijeps}\ref{item:welldef sigmaijeps 2} we infer that $X_i\mathsf{R}$ has finite trace per unit volume for any $i\in\{1,\dots,d\}$ and equality~\eqref{eqn:tau reads only prim cell} is implied by Proposition~\ref{prop:AXi}\ref{it:AXi exh dep}.

Finally, Lemma~\ref{prop:cycl of tau} implies that $\tau(\mathsf{E}_2)=0=\tau(\mathsf{E}_4)$. 
Indeed, for the first identity observe that 
\[
\tau(  X_i\su{OD}S\Pi_1H_0)=\tau( \Pi_1H_0 X_i\su{OD}S)=\tau( H_0 X_i\su{OD}S\Pi_1),
\]
and for the second one just notices that $\mathsf{E}_4$ is the commutator of $\overline{[X_i,\Pi_0]}\in\B_\infty^\tau$ and $\Pi_0 S \overline{[\Pi_0,X_j]}\in\B_1^\tau$.
\end{proof}

\begin{proof}[Proof of Proposition \ref{prop:Btorque}]
{\ref{item:doublecomm formula}}
First of all, notice that $[H_0,B]\Pi_1\in \B_1^\tau$ since $\Pi_1=\Pi_1\Pi_0+\Pi_0\Pi_1$ and, for  $\Pi_\sharp \in 
\set{\Pi_0, \Pi_1}$, one has  $[H_0,B]\cdot\Pi_\sharp\in\PO(\Df,\Hf)\cdot\PO(\Hf,\Df)\subset \PO(\Hf)$ and $\Pi_\sharp\in\PO(\Hf,\Df)\cap\B_1^\tau$.  
Since 
$$
\Pi_1=\frac{\iu}{2 \pi} \oint_C \di z \, (H_0 - z \Id)^{-1} \, \big[\,\overline{[X_j,\Pi_0]},\Pi_0\,\big] \, (H_0 - z \Id)^{-1}
$$ by construction (see Propositions~\ref{prop:Pi1}\ref{item:propSA3} and \ref{prop:I(A)}), we have that
\begin{align}
\iu\tau([H_0,B]\Pi_1)&=-\frac{1}{2\pi}\oint_C \di z \,\tau\([H_0,B] (H_0 - z \Id)^{-1} \, 
\big[\,\overline{[X_j,\Pi_0]},\Pi_0\,\big] \, (H_0 - z \Id)^{-1}\)\cr
&=-\frac{1}{2\pi}\oint_C \di z \,\tau\((H_0 - z \Id)^{-1} [H_0,B] (H_0 - z \Id)^{-1} \big[\,\overline{[X_j,\Pi_0]},\Pi_0\,\big]\)\cr
&=-\frac{1}{2\pi}\oint_C \di z \,\tau\( [B,(H_0 - z \Id)^{-1}] \big[\,\overline{[X_j,\Pi_0]},\Pi_0\,\big]\)\cr
\label{eqn:interm formula for doublecomm}
&=\iu\tau\( [B,\Pi_0] \big[\,\overline{[X_j,\Pi_0]},\Pi_0\,\big]\),
\end{align}
where in the first equality we have used the cyclicity of the trace per unit volume , based on the fact 
that $[H_0,B] (H_0 - z \Id)^{-1}\in \B_\infty^\tau$ by Lemma~\ref{lem:resolvent-projectionPHfDf}, 
and that $\big[\,\overline{[X_j,\Pi_0]},\Pi_0\,\big]=\overline{[X_j,\Pi_0]}\cdot \Pi_0-\Pi_0\cdot \overline{[X_j,\Pi_0]}\in\PO(\Hf)\cdot\B_1^\tau+\B_1^\tau\cdot \PO(\Hf)\subset\B_1^\tau$ by Corollary~\ref{cor:der Pi0 and H0}\ref{item:[Pi0,X_j]} and Lemma~\ref{lem:Pi in tau-class}.

Finally, in virtue of $[B,\Pi_0]\overline{[X_j,\Pi_0]}\Pi_0=\Pi_0 [B,\Pi_0]\overline{[X_j,\Pi_0]}$ and by using again the cyclicity of the trace per unit volume since $\Pi_0\overline{[X_j,\Pi_0]}\in\B_1^\tau$ and $[B,\Pi_0]\in \B_\infty^\tau$,
we conclude that
\begin{align*}
\tau\( \iu[B,\Pi_0] \big[\,\overline{[X_j,\Pi_0]},\Pi_0\,\big]\)&=\tau\(\iu [B,\Pi_0]\overline{[X_j,\Pi_0]}\Pi_0\)
-\tau\(\iu [B,\Pi_0]\Pi_0\overline{[X_j,\Pi_0]}\)\\
&=\tau\(\iu \Pi_0 [B,\Pi_0]\overline{[X_j,\Pi_0]}\)-\tau\(\iu \Pi_0\overline{[X_j,\Pi_0]} [B,\Pi_0]\)\\
&=\tau\(\iu \Pi_0 \big[\,[B,\Pi_0],\overline{[X_j,\Pi_0]} \, \big] \).
\end{align*}

\noindent {\ref{item: vanishing of doublecomm formula if A per}}
In view of intermediate formula \eqref{eqn:interm formula for doublecomm}, the claim is equivalent to show that 
$$\tau\( [B,\Pi_0] \big[\,\overline{[X_j,\Pi_0]},\Pi_0\,\big]\)=0.$$
By algebraic manipulations, exploiting the fact that $\Pi_0$ is a projection, we obtain that on $\Ran(\chi_1)$
\begin{align*}
[B,\Pi_0] \big[\,\overline{[X_j,\Pi_0]},\Pi_0\,\big]&=\Pi_0^\perp B\Pi_0 X_j\Pi_0^\perp-\Pi_0 B\Pi_0^\perp X_j\Pi_0\\
&=\Pi_0^\perp B\Pi_0\overline{[\Pi_0,X_j]}+\Pi_0 B\Pi_0^\perp\overline{[\Pi_0,X_j]}.
\end{align*}
Therefore, since each summand on the right-hand side above is in $\B_1^\tau$, using again the cyclicity of the trace per unit volume and the off-diagonality of $\overline{[\Pi_0,X_j]}$, we get that

\begin{align*}
\tau([B,\Pi_0] \big[\,\overline{[X_j,\Pi_0]},\Pi_0\,\big])&=\tau(\Pi_0^\perp B\Pi_0\overline{[\Pi_0,X_j]})+\tau(\Pi_0 B\Pi_0^\perp\overline{[\Pi_0,X_j]})\\
&=\tau( B\Pi_0\overline{[\Pi_0,X_j]}\Pi_0^\perp)+\tau(B\Pi_0^\perp\overline{[\Pi_0,X_j]}\Pi_0)\\
&=\tau(B\overline{[\Pi_0,X_j]})=\frac{1}{\abs{\FC_1}}\Tr\left( \chi_1 {[B\Pi_0,X_j]} \chi_1 \right)\\
&=\frac{1}{\abs{\FC_1}}\Big\{\!\! \Tr\left( \chi_1 B\Pi_0 \chi_1 X_j \chi_1 \right)-\Tr\left( \chi_1 X_j \chi_1 B\Pi_0  \chi_1  \right)\!\!\Big\}=0,
\end{align*}
where we have used that $[B,X_j]=0$ and the conditional cyclicity of the trace, since $\chi_1 B\Pi_0 \chi_1$ is trace class in view of the fact that $B\Pi_0\in\B_1^\tau$ and of Lemma~\ref{lem:tau cont funct1andL}.
\end{proof}

\begin{proof}[Proof of Theorem~\ref{thm:notKubo}]
First of all, Proposition~\ref{cor:welldef sigmaijeps}\ref{item:welldef sigmaijeps 2} implies that $J\subm{prop}{i}^S\Pi_\sharp$ has a finite trace per unit volume for $\Pi_\sharp\in\{ \Pi_0,\Pi_1,\Pi_r^\eps\}$, thus all the terms appearing in \eqref{eqn:sigmaprop} are finite. By virtue of Proposition~\ref{cor:welldef sigmaij}, we obtain that 
$\sigma\subm{prop}{ij}^S = \re\tau(J\subm{prop}{i}^S\,\Pi_1)$. 
Now, notice that
$\re\tau(\iu[H_0,X_i]S\,\Pi_1) =\tau(J\subm{conv}{i}^S\,\Pi_1)$ by Lemma~\ref{Lem:reduce_Jconv}.
Therefore, by previous computation in \eqref{eqn:Ji0Pi1chiL} and \eqref{eq:T+R}, we get that
\begin{align}
\re\tau(J\subm{prop}{i}^S\,\Pi_1)&=\re\tau(\iu[H_0,X_i]S\,\Pi_1) + \re\tau(X_i\mathsf{R})\cr
&=\tau(J\subm{conv}{i}^S\,\Pi_1)+\re\tau(X_i\mathsf{R}).
\end{align}
Observe that Proposition~\ref{cor:welldef sigmaij} implies that 
$\sigma\subm{conv}{ij}^S = \re\tau(J\subm{conv}{i}^S\,\Pi_1) = \tau(J\subm{conv}{i}^S\,\Pi_1)$. 
From equation \eqref{eqn:K and beyond} and Proposition~\ref{prop:Extrachanged} we derive formula \eqref{eqn:sigmaT}. Corollary~\ref{cor:tau R=0} along with Proposition~\ref{prop:AXi}\ref{it:AXi exh indep} imply the well-posedness of the rotation $S$-conductivity, as it appears in \eqref{eqn:sigmaR}.
\end{proof}

\begin{proof}[Proof of Lemma~\ref{lemma:Extra2}]
Obviously, $\mathsf{E}_3=0$ and $\mathsf{R}=0$ by using \eqref{Sconserved}. Proposition~\ref{prop:Extrachanged} implies that $\tau(\mathsf{E}_1)=\iu\tau(  [H_0,X_i\su{D}] S\su{OD} \Pi_1)=0$, since $S\su{OD}=0$.
\end{proof}

\appendix

\section{Unit Cell Consistency and vanishing of persistent $S$-currents} \label{app:Persistent}

\subsection{Results on the Unit Cell Consistency}
\label{Sec:UnitCellConsistency}

As a preliminary step, we prove that the trace per unit volume, acting on a 
suitable class of operators, is independent of the choice of the fundamental cell $\FC_1$.
First, notice that chosen any two primitive cells of arbitrary shape, it is possible to cut the first up into pieces, which, when translated through suitable lattice vectors, can be reassembled to give the second. This fact, well-known to solid state physicists \cite{AshcroftMermin}, can be reformulated in the following Lemma. 
Recall that $\FC_L$ is defined in \eqref{eqn:defn FCL} with reference to a linear basis $\set{a_1, \ldots, a_d}$ of $\Gamma$, while $\widetilde{\FC}_L$ in \eqref{eqn:defn tildeFCL} refers to another linear basis $\set{\tilde{a}_1, \ldots, \tilde{a}_d}$ of $\Gamma$.
 
\begin{lemma}
Let $\FC_1$ and $\widetilde{\FC}_1$ be two primitive cells. Then there exist a finite set $I\subset\Gamma$ and a family of subsets $\set{P_\gamma}_{\gamma\in I}\subset \X$ such that$^{\ref{fn:disjun}}$
\begin{equation}
\label{eq:decomps C_1 and tildeC_1}
\FC_1= \bigsqcup_{\gamma\in I} \mathrm{T}_{\gamma} P_\gamma \quad\text{ and }\quad\widetilde{\FC}_1=\bigsqcup_{\gamma\in I}  P_\gamma.
\end{equation}
\end{lemma}
\begin{proof}
Clearly, by the very definition of a primitive cell we have that 
\begin{equation}
\label{eqn:spaceastranslofprcell}
\bigsqcup_{\gamma\in\Gamma} \mathrm{T}_{\gamma}\FC_1=\X=\bigsqcup_{\gamma\in\Gamma} \mathrm{T}_{\gamma}\widetilde{\FC}_1.
\end{equation}
Therefore, by the second identity, we can rewrite $\FC_1$ as
\begin{align*}
\FC_1=\bigsqcup_{\gamma\in\Gamma} \FC_1\cap \mathrm{T}_{\gamma}\widetilde{\FC}_1=\bigsqcup_{\gamma\in I} \FC_1\cap \mathrm{T}_{\gamma}\widetilde{\FC}_1=\bigsqcup_{\gamma\in I}\mathrm{T}_{\gamma} \(\mathrm{T}_{-\gamma}\FC_1\cap \widetilde{\FC}_1\),
\end{align*}
where we have used that, since $\FC_1$ is compact, there exists a finite subset $I\subset \Gamma$ such that $\FC_1\cap \mathrm{T}_{\gamma}\widetilde{\FC}_1= \varnothing$ any $\gamma\in \Gamma \setminus I$. We set 
\begin{equation} \label{P_gamma}
P_\gamma:=\mathrm{T}_{-\gamma}\FC_1\cap \widetilde{\FC}_1\subset \widetilde{\FC}_1. 
\end{equation}
The proof is concluded by observing that
\begin{align*}
\bigsqcup_{\gamma\in I}P_\gamma=\bigsqcup_{\gamma\in \Gamma}P_\gamma=\bigsqcup_{\gamma\in \Gamma}\mathrm{T}_{-\gamma}\FC_1\cap \widetilde{\FC}_1=\widetilde{\FC}_1,
\end{align*}
where we have used the first identity in \eqref{eqn:spaceastranslofprcell}.
\end{proof}

Denoting by $\tau(\,\cdot\,)$ and $\widetilde{\tau}(\,\cdot\,)$, respectively, the trace per unit volume induced by the choice of the primitive cells $\FC_1$ and $\widetilde{\FC}_1$, we have the following result.

\begin{proposition}
\label{cor:tauindprcell}
Consider an operator $A$ which is periodic and trace class on compact sets. 
\begin{enumerate}[label=(\roman*), ref=(\roman*)]
\item \label{it: tau indep for Aperiodic} Then $ \tau(A)=\widetilde{\tau}(A)$.
\item \label{it: tau dep for X_iA}
If, in addition, $\Tr(\chi_{P_\gamma} A\, \chi_{P_\gamma})=0$ for all $\gamma\in I$, 
where $\set{P_\gamma}_{\gamma \in I}$ are the sets defined in \eqref{P_gamma},
then $\tau(X_iA)=\widetilde{\tau}(X_iA)$.
\end{enumerate}
\end{proposition}
\begin{proof}
\ref{it: tau indep for Aperiodic}
In view of Proposition~\ref{prop:charge-tauPeriodic}\ref{item:per+traceclassoncompact}, it suffices to prove
\[
\frac{1}{\abs{\FC_1}} \Tr(\chi_1\, A\, \chi_1)=\frac{1}{ |\widetilde{\FC}_1| } \Tr(\widetilde{\chi}_1\, A\, \widetilde{\chi}_1).
\]

Obviously, from decompositions \eqref{eq:decomps C_1 and tildeC_1} and the translation invariance of the Lebesgue measure, it follows that $\abs{\FC_1}=|\widetilde{\FC}_1|$. The first identity in \eqref{eq:decomps C_1 and tildeC_1}, conditional cyclicity of the trace and identity $\chi_{\mathrm{T}_{\gamma} P_\gamma}^2=\chi_{\mathrm{T}_{\gamma} P_\gamma}$ imply that
\begin{equation} \label{eqn:pieces}
\Tr(\chi_1\, A\, \chi_1)=\sum_{\gamma\in I}\Tr(\chi_{\mathrm{T}_{\gamma} P_\gamma}\, A\, \chi_{\mathrm{T}_{\gamma} P_\gamma})=\sum_{\gamma\in I}\Tr(T_{\gamma}\chi_{P_\gamma}T^*_{\gamma}\, A\, T_{\gamma}\chi_{P_{\gamma}}T^*_\gamma).
\end{equation}
Now, by the invariance of trace under unitary conjugation and the periodicity of $A$, one has that 
\[
\Tr(\chi_1\, A\, \chi_1)=\sum_{\gamma\in I}\Tr(\chi_{P_\gamma}\, A\, \chi_{P_\gamma})=\Tr(\widetilde{\chi}_1\, A\, \widetilde{\chi}_1),
\]
where in the last equality we used the second decomposition in \eqref{eq:decomps C_1 and tildeC_1}.

\ref{it: tau dep for X_iA} After arguing as in the steps leading to \eqref{eqn:pieces}, one notices that
\begin{align*}
\Tr(T_{\gamma}\chi_{P_\gamma}T^*_{\gamma}\, X_i A\, T_{\gamma}\chi_{P_\gamma}T^*_{\gamma})&=\Tr(\chi_{P_\gamma}T^*_{\gamma}\, X_i A\, T_{\gamma}\chi_{P_\gamma})=\Tr(\chi_{P_\gamma}T^*_{\gamma}\, X_i T_{\gamma} A\, \chi_{P_\gamma})\\
&=\gamma_i\Tr(\chi_{P_\gamma} A\, \chi_{P_\gamma})+\Tr(\chi_{P_\gamma} X_i A\, \chi_{P_\gamma}) \\
& = \Tr(\chi_{P_\gamma} X_i A\, \chi_{P_\gamma})
\end{align*}
where we have used $[T_\gamma^*,X_i]=\gamma_i \,T_\gamma^*$. Therefore, in view of the decompositions  \eqref{eq:decomps C_1 and tildeC_1}, we conclude that
\begin{align*}
\Tr(\chi_1\, X_i A\, \chi_1)=\sum_{\gamma\in I}\gamma_i\Tr(\chi_{P_\gamma} A\, \chi_{P_\gamma})+\Tr(\widetilde{\chi}_1\, X_i A\, \widetilde{\chi}_1)=\Tr(\widetilde{\chi}_1\, X_i A\, \widetilde{\chi}_1).
\end{align*}
which yields the claim. 
\end{proof}

For any $\gamma\in I$, consider the operator  $\Id_{\set{T_\eta T_\gamma P_\gamma\,:\, \eta \in \Gamma}}=\Id_{\set{T_\eta P_\gamma\,:\, \eta \in \Gamma}}$, which is periodic by its very definition.  
(Here $\Id_\Omega$ is an alternative notation for the indicatrix function of the set $\Omega$).
By applying \eqref{eqn:tauPeriodic1}, one has
\begin{equation}
\label{eqn:rewr of mesopointwise spin torque}
|{\FC}_1|\tau(\Id_{\set{T_\eta P_\gamma\,:\, \eta \in \Gamma}}A)= \Tr({\chi}_1\Id_{\set{T_\eta P_\gamma\,:\, \eta \in \Gamma}} A {\chi}_1)=\Tr(\chi_{P_\gamma}\, A\, \chi_{P_\gamma}),
\end{equation}
for every operator $A$ which is periodic and trace class on compact sets. Using the previous rewriting, we are in position to prove that the restriction of the $S$-rotation part $\chi_{P_\gamma}\, \mathsf{R}\, \chi_{P_\gamma}$, defined in \eqref{eqn:defn T and R}, has vanishing trace whenever the model enjoys a discrete rotational symmetry.

\subsection{Models with discrete rotational symmetries}
Let us fix indices $i\neq j \in \set{1,\ldots,d}$, and denote by $\mathrm{R}_{\vartheta,(ij)}$ the counterclockwise rotation of angle $\vartheta\in [0,2\pi)$ in the plane $(x_i,x_j)$:
\begin{multline*}
\mathrm{R}_{\vartheta,(ij)}(x_1,\ldots,x_i,\ldots,x_j,\ldots,x_d) \\
:= (x_1,\ldots,(\cos\vartheta) x_i - (\sin\vartheta) x_j,\ldots,(\sin \vartheta) x_i +(\cos\vartheta) x_j,\ldots,x_d).
\end{multline*}
\emph{Rotation operators in the plane $(x_i,x_j)$} on $\Hi$ are defined via
\[
(R_{\vartheta,(ij)}\psi)(x):=\rho_{\vartheta,(ij)}\psi(\mathrm{R}_{\vartheta,(ij)}^{-1}x), \quad \text{for } \psi\in\Hi,
\]
where $\rho_{\vartheta,(ij)}$ is a unitary matrix acting on $\C^N$\footnote{In $2$-level systems one defines $\rho_{\vartheta,(12)}:=\E^{-\iu \vartheta s_z}$ to encode the rotation of angle $\theta$ around the $z$-axis on $\C^2$.}.

Suppose that  the $d$-dimensional crystal under consideration is invariant under a rotation of angle $\vartheta = 2\pi/n$, for some $n \in \N^*$, in the plane $(x_i,x_j)$, namely $\gamma\in\Gamma$ if and only if $\mathrm{R}_{\vartheta, (ij)} \gamma\in\Gamma$ (then it trivially follows that $x\in\X$ if and only if $\mathrm{R}_{\vartheta, (ij)} x\in\X$).
A periodic Hamiltonian $H_0$ is said to be itself \emph{rotationally symmetric} or \emph{invariant under rotation of angle $\vartheta$  in the plane $(x_i,x_j)$} if and only if $R_{\vartheta,(ij)}^{-1} \, H_0 \, R_{\vartheta,(ij)} = H_0$. For example, several models on the honeycomb structure, including \eg the Kane--Mele model (see \cite{KaneMele2005a} or \cite[Appendix A]{MarcelliPanatiTauber18}), are invariant under the rotation $R_{2\pi/3, (12)}$.

\begin{proposition}
\label{prop:van mesopointwise spin torque}
Let $\vartheta=2\pi/n$ for some $n\in\N^*$. 
Let the Bravais lattice $\Gamma$ be invariant under the rotation $\mathrm{R}_{\vartheta,(ij)}$, \ie $\gamma\in\Gamma$ if and only if $\mathrm{R}_{\vartheta, (ij)} \gamma\in\Gamma$, and
\begin{equation}
\label{eqn:subclassofmodel}
R_{\vartheta, (ij)}^{-1}\Id_{\set{T_\eta P_\gamma\,:\, \eta \in \Gamma}}R_{\vartheta, (ij)}=\Id_{\set{T_\eta P_\gamma\,:\, \eta \in \Gamma}} 
\qquad \forall\gamma\in I.
\end{equation}
Let the operator $H_0$, as in Assumption~\ref{assum:1}, be rotationally symmetric of angle $\vartheta$  in the plane $(x_i,x_j)$. Let $S= \Id_{L^2(\X)} \otimes s$, as in Definition~\ref{def:$S$-current}, be such that $\rho_{\vartheta,(ij)}^{-1}s\rho_{\vartheta,(ij)}=s$. Then

\begin{enumerate}[label=(\roman*), ref=(\roman*)]
\item \label{item:pointwisemesospintorq}
\[
\Tr(\chi_{P_\gamma}\, \mathsf{R} \,\chi_{P_\gamma})=0  \qquad \forall\gamma\in I,
\]
where $\mathsf{R}=\iu  [H_0, S ]  \Pi_1$ and 
$\set{P_\gamma}_{\gamma \in I}$ are the sets defined in \eqref{P_gamma}.

\item \label{item:persistentcurrent} the persistent \emph{conventional} $S$-current vanishes, namely
$\tau(J\subm{conv}{i}^S\Pi_0)=0.$ 
\end{enumerate}
\end{proposition}

\begin{proof}
\ref{item:pointwisemesospintorq}
In view of Proposition~\ref{cor:welldef sigmaijeps}\ref{item:welldef sigmaijeps 1} $\Id_{\set{T_\eta P_\gamma\,:\, \eta \in \Gamma}} \mathsf{R}\in\B_1^\tau$ and thus applying \eqref{eqn:rewr of mesopointwise spin torque} the thesis is equivalent to show that 
$$
\tau(\Id_{\set{T_\eta P_\gamma\,:\, \eta \in \Gamma}} \mathsf{R} )=\iu\tau(\Id_{\set{T_\eta P_\gamma\,:\, \eta \in \Gamma}}   [H_0, S ]  \Pi_1 )=0.
$$
Using the invariance of the trace under unitary conjugation and the identities $R_{\vartheta, (ij)}^{-1}\chi_1 R_{\vartheta, (ij)}=\widetilde{\chi}_1$, $R_{\vartheta, (ij)}^{-1} H_0 R_{\vartheta, (ij)}=H_0$ and $R_{\vartheta, (ij)}^{-1} S R_{\vartheta, (ij)}=S$ and \eqref{eqn:subclassofmodel}, and $ \Pi_1 = \mathcal{I}\( \overline{[  X_j,\Pi_0]}\) $ by Proposition~\ref{prop:Pi1}\ref{item:propSA3}, we obtain that
\begin{align*}
&|{\FC}_1|\tau(\Id_{\set{T_\eta P_\gamma\,:\, \eta \in \Gamma}}   [H_0, S ]  \mathcal{I}\( {[  X_j,\Pi_0]}\) )\\
&\quad=\Tr({\chi}_1\Id_{\set{T_\eta P_\gamma\,:\, \eta \in \Gamma}}   [H_0, S ]  \mathcal{I}\( {[  X_j,\Pi_0]}\) {\chi}_1)\\
&\quad=\Tr(\widetilde{\chi}_1\Id_{\set{T_\eta P_\gamma\,:\, \eta \in \Gamma}}   [H_0, S ]  \mathcal{I}\( {[ R_{\vartheta, (ij)}^{-1} X_j R_{\vartheta, (ij)},\Pi_0]}\) \widetilde{\chi}_1)\\
&\quad=|\widetilde{\FC}_1|\widetilde{\tau}\(\Id_{\set{T_\eta P_\gamma\,:\, \eta \in \Gamma}}  [H_0, S ]  \mathcal{I}\( {[ R_{\vartheta, (ij)}^{-1} X_j R_{\vartheta, (ij)},\Pi_0]}\)\)\\
&\quad=|{\FC}_1|{\tau}\(\Id_{\set{T_\eta P_\gamma\,:\, \eta \in \Gamma}}  [H_0, S ]  \mathcal{I}\( {[ R_{\vartheta, (ij)}^{-1} X_j R_{\vartheta, (ij)},\Pi_0]}\)\),
\end{align*}
where we have used Proposition~\ref{prop:charge-tauPeriodic}\ref{item:per+traceclassoncompact} and Proposition~\ref{cor:tauindprcell}\ref{it: tau indep for Aperiodic}. Therefore, by iterating the previous computation we have that 
\begin{align}
&\tau(\Id_{\set{T_\eta P_\gamma\,:\, \eta \in \Gamma}}   [H_0, S ]  \mathcal{I}\( {[  X_j,\Pi_0]}\) )\cr
\label{eqn:spintorquek}
&\quad=\frac{1}{n}\sum_{k=0}^{n-1}{\tau}\(\Id_{\set{T_\eta P_\gamma\,:\, \eta \in \Gamma}}  [H_0, S ]  \mathcal{I}\( {[ R_{\vartheta, (ij)}^{-k} X_j R_{\vartheta, (ij)}^{k},\Pi_0]}\)\).
\end{align}
Now we are going to compute $\sum_{k=0}^{n-1} R_{\vartheta, (ij)}^{-k} X_j
 R_{\vartheta, (ij)}^{k}$.
The rotation of angle $\vartheta$ acts non-trivially only in the plane $(x_i,x_j)$, which we parametrize with the complex coordinate $z := x_i + \iu x_j$. In this parametrization, the rotation of angle $\vartheta$ is implemented as $\hat{\mathrm{R}}_{\vartheta}z:= \E^{\iu \vartheta}z$.  Introducing the complex position operator $Z:=X_i+\iu X_j$, one has then that $\hat{R}_{\vartheta, (ij)}^{-k}\,\iu X_j\,\hat{R}_{\vartheta, (ij)}^{k}=\im\(\E^{\iu k\vartheta}Z\)$ and thus 
\[  \sum_{k=0}^{n-1} \hat{R}_{\vartheta, (ij)}^{-k} \,\iu X_j\, \hat{R}_{\vartheta, (ij)}^{k} =\im \left( \sum_{k=0}^{n-1} \E^{\iu 2 \pi k/n}  Z\right). \]
As $\sum_{k=0}^{n-1}\E^{\iu 2\pi k/n}=0$, we deduce that the term in \eqref{eqn:spintorquek} vanishes. This concludes the proof.

\ref{item:persistentcurrent} In view of decomposition \eqref{eq:decomps C_1 and tildeC_1}, it suffices to show that 
\[
\Tr(\chi_{P_\gamma}\,i[H_0,X_i]S\Pi_0  \,\chi_{P_\gamma})=0\quad \forall\gamma\in I,
\]
whose proof is analogous to the previous one since $[H_0,X_i]S\Pi_0\in\B_1^\tau.$
\end{proof}

\begin{remark}
In general, even exploiting the peculiar discrete rotational symmetries in the hypotheses of Proposition~\ref{prop:van mesopointwise spin torque}, it is not obvious that the persistent \emph{proper} $S$-current vanishes, \ie $\tau( J\subm{prop}{i}^S\Pi_0)=0$, since the argument we used relies on the periodicity of the operators involved and $J\subm{prop}{i}^S$ is not periodic. Nevertheless, in the Kane--Mele model this property holds true thanks to the specific structure of the model.
\end{remark}

\subsection{Vanishing of persistent $S$-current when $S$ is conserved} \label{app:persistentS}

When $S$ is a conserved quantity, namely when $[H_0, S]=0$, the vanishing of the persistent $S$-current $J_i^S=J\subm{prop}{i}^S=J\subm{conv}{i}^S$ holds true without any symmetry assumption on $H_0$ (compare \cite{Bellissard94, BoucletGerminetKleinSchenker05}, where similar results are deduced in the case $S = \Id$, and \cite{Bachmann_et_al18}, which offers a proof in the context of many-body quantum spin systems).

To show this, notice first that $[\Pi_0,S]=0$ as well, and that $J_i^S = \iu \overline{[H_0,  X_i]}\, S$ is a periodic operator in view of Lemma~\ref{lemma:derivata}. 
Moreover, $J_i^S \Pi_0=\iu \overline{[H_0,  X_i]}\, S\Pi_0$ is $\tau$-class in view of Proposition~\ref{cor:welldef sigmaijeps}\ref{item:welldef sigmaijeps 1}. Consequently, by the identity $\Pi_0^2=\Pi_0$, Lemma~\ref{prop:cycl of tau} and Proposition~\ref{prop:charge-tauPeriodic}\ref{item:per+traceclassoncompact}, we have that
\begin{align*}
\tau(\overline{[H_0,  X_i]} S\Pi_0)&=\tau(\Pi_0\overline{[H_0,  X_i]} S\Pi_0)=\frac{1}{|\FC_1|}\Tr(\chi_1\Pi_0\overline{[H_0,  X_i]} S\Pi_0\chi_1).
\end{align*}
By Remark~\ref{rmk:SmoothBF}, 
after some algebra, we get that
\begin{align}
\chi_1\Pi_0\overline{[H_0,  X_i]} S\Pi_0\chi_1&=\chi_1\Pi_0[H_0,  X_i] S\Pi_0\chi_1=\chi_1[\Pi_0 H_0 S \Pi_0, \Pi_0 X_i\Pi_0] \chi_1\cr
\label{eqn:perscurrent two terms}
&=\chi_1[\Pi_0 H_0 S \Pi_0, X_i] \chi_1-\chi_1[\Pi_0 H_0 S \Pi_0, {X_i}\su{OD}] \chi_1.
\end{align}
Notice that the trace of the first summand above vanishes:
\begin{align*}
\Tr(\chi_1[\Pi_0 H_0 S \Pi_0, X_i] \chi_1)&=\Tr(\chi_1\Pi_0 H_0 S \Pi_0\chi_1 X_i \chi_1 - \chi_1 X_i \chi_1\Pi_0 H_0 S \Pi_0 \chi_1),
\end{align*}
where both summands inside the trace are trace class. Indeed, $\Pi_0 \chi_1$ is an Hilbert--Schmidt operator, since $\Tr((\Pi_0 \chi_1)^* \Pi_0 \chi_1) = \Tr(\chi_1 \Pi_0 \chi_1) = |\FC_1| \tau(\Pi_0) < \infty$ by Lemma~\ref{lem:Pi in tau-class}. This implies that the adjoint $\chi_1 \Pi_0$ is Hilbert--Schmidt as well. Since $S$, $\Pi_0 H_0 \Pi_0$ and $\chi_1 X_i \chi_1$ are all bounded operators, the desired claim follows in view of the conditional cyclicity of the trace $\Tr(\,\cdot\,)$. 
Finally, we have that the trace of the second summand in \eqref{eqn:perscurrent two terms} vanishes as well. Indeed, by Remark~\ref{rmk:SmoothBF}, Lemma~\ref{lemma:diag}\ref{item:diag2} and definition~\eqref{eqn:closure der of Pi0 and H0} we obtain that
\[
\chi_1[\Pi_0 H_0 S \Pi_0, {X_i}\su{OD}] \chi_1=\chi_1[\Pi_0 H_0 S \Pi_0, [\,\overline{[X_i,\Pi_0]}\,,\Pi_0]] \chi_1
\]
and thus, using Proposition~\ref{prop:charge-tauPeriodic}\ref{item:per+traceclassoncompact}, we deduce
\begin{align*}
\frac{1}{|\FC_1|}\Tr(\chi_1[\Pi_0 H_0 S \Pi_0, {X_i}\su{OD}] \chi_1)=& \tau\(\Pi_0 H_0 S \Pi_0 [\,\overline{[X_i,\Pi_0]}\,,\Pi_0]\)+\\
&-\tau\([\,\overline{[X_i,\Pi_0]}\,,\Pi_0] \Pi_0 H_0 S \Pi_0 \).
\end{align*}
The conditional cyclicity of $\tau$ implies the conclusion, since $\Pi_0 H_0 S \Pi_0\in\B_\infty^\tau$ and $[\,\overline{[X_i,\Pi_0]}\,,\Pi_0]\in \B_1^\tau$.


\bigskip \bigskip


{\tiny  

\begin{tabular}{ll}

(G.~Marcelli) & \textsc{Fachbereich Mathematik, Eberhard Karls Universit\"{a}t T\"{u}bingen} \\
 &  Auf der Morgenstelle 10, 72076 T\"{u}bingen, Germany \\
 &  {E-mail address}: \href{mailto:giovanna.marcelli@uni-tuebingen.de}{\texttt{giovanna.marcelli@uni-tuebingen.de}}\\
\\
(G.~Panati) & \textsc{Dipartimento di Matematica, ``La Sapienza'' Universit\`{a} di Roma} \\
 &  Piazzale Aldo Moro 2, 00185 Rome, Italy \\
 &  {E-mail address}: \href{mailto:panati@mat.uniroma1.it}{\texttt{panati@mat.uniroma1.it}} \\
\\
(S.~Teufel) & \textsc{Fachbereich Mathematik, Eberhard Karls Universit\"{a}t T\"{u}bingen} \\
 &  Auf der Morgenstelle 10, 72076 T\"{u}bingen, Germany \\
 &  {E-mail address}: \href{mailto:stefan.teufel@uni-tuebingen.de}{\texttt{stefan.teufel@uni-tuebingen.de}} \\
 \\

\end{tabular}

}
\end{document}